\documentclass[lettersize,journal]{IEEEtran}
\IEEEoverridecommandlockouts
\ifCLASSINFOpdf
  \usepackage[pdftex]{graphicx}
\else
  \usepackage[dvips]{graphicx}
\fi

\usepackage[hyphens]{url}
\usepackage[cmex10]{amsmath}
\usepackage{bbm}
\usepackage{amssymb}
\usepackage{xcolor,colortbl}
\usepackage{tablefootnote}
\usepackage{graphicx}
\usepackage{float}
\usepackage[caption=false,font=footnotesize]{subfig}
\usepackage{wrapfig}
\usepackage{algorithmic}
\usepackage[algoruled, linesnumbered]{algorithm2e} 
\usepackage{mathrsfs}
\usepackage{cite}
\usepackage{array}
\usepackage{mdwmath}
\usepackage{mdwtab}
\usepackage{fixltx2e}
\usepackage{amsthm}
\usepackage{tikz}
\allowdisplaybreaks

\newtheorem{lemma}{Lemma}

\newtheorem{remark}{Remark}
\DeclareMathOperator*{\argmax}{argmax}
\DeclareMathOperator*{\argmin}{argmin}

\begin{document}
%
\title{Optimal preprocessing of WiFi CSI for sensing applications}
\author{Vishnu V. Ratnam, \IEEEmembership{Senior Member,~IEEE}, Hao Chen, \IEEEmembership{Member,~IEEE}, Hao Hsuan Chang, \IEEEmembership{Member,~IEEE}, \\ Abhishek Sehgal, \IEEEmembership{Member,~IEEE}, Jianzhong (Charlie) Zhang, \IEEEmembership{Fellow,~IEEE}
\thanks{All authors are with the Standards and Mobility Innovation Lab, Samsung Research America, Plano, Texas, USA. Email: ratnamvishnuvardhan@gmail.com}}
\maketitle

\begin{abstract}
Due to its ubiquitous and contact-free nature, the use of WiFi infrastructure for performing sensing tasks has tremendous potential. However, the channel state information (CSI) measured by a WiFi receiver suffers from errors in both its gain and phase, which can significantly hinder sensing tasks. By analyzing these errors from different WiFi receivers, a mathematical model for these gain and phase errors is developed in this work. Based on these models, several theoretically justified preprocessing algorithms for correcting such errors at a receiver and, thus, obtaining clean CSI are presented. Simulation results show that at typical system parameters, the developed algorithms for cleaning CSI can reduce noise by $40$\% and $200$\%, respectively, compared to baseline methods for gain correction and phase correction, without significantly impacting computational cost. The superiority of the proposed methods is also validated in a real-world test bed for respiration rate monitoring (an example sensing task), where they improve the estimation signal-to-noise ratio by $20$\% compared to baseline methods.
\end{abstract}

\begin{IEEEkeywords}
WiFi sensing, Wireless sensing, CSI preprocessing, Respiration rate estimation, Smart home.
\end{IEEEkeywords}

\section{Introduction} \label{sec_intro}
With the advent of the internet of things and the ubiquitous availability of WiFi infrastructure (access points and stations), WiFi-based wireless sensing has become a very hot and upcoming topic \cite{Jiang2018, Ma2019}. In wireless sensing, a transmitter (TX) periodically transmits a known signal, and the receiver (RX) uses the received signal to track temporal variations in the channel and, correspondingly, in the ambient environment. This has many applications including presence detection, exercise monitoring, people counting, intruder alarm, respiration rate detection, sleep monitoring, etc. \cite{Fang2018, Charlton2018, Regev2021, Natarajan2021}. 

A typical difficulty of WiFi devices is that the TX and RX do not have timing and carrier synchronization, due to which the channel estimate at the RX, referred to as channel state information (CSI), suffers from random phase errors. In addition, the impact of variable gain circuits in the RX causes random amplitude fluctuations in the channel estimate at the RX. Since these gain and phase errors are not easily separable from the variations in the ``true" channel (which captures the variations in the environment relevant to sensing), wireless sensing using WiFi devices is very challenging. Prior works have devised several approaches to tackle these errors. 

\subsubsection{Accounting for gain errors} In a first approach, metrics robust to gain errors are used for sensing. Among them, \cite{Wang2017b, Wang2017c} only CSI phase information for sensing. Some works assume the gain error across RX antennas to be correlated and thus use ratio of CSI at two antennas as the metric \cite{Zeng2019}. 
In a second approach, the gain errors are explicitly estimated and removed. The most popular one assumes the CSI power variations to only be caused by gain errors and uses the square root of the CSI power as an estimate of the variable gain \cite{Chen2018, Hu2022}. In \cite{Liu2021}, a DBSCAN-based clustering algorithm is proposed to identify CSI values belonging to different discrete gain levels, and uses the square root of the cluster centers as an estimate of the variable gain. In \cite{Niu2021}, a method to use side information such as RSSI and gain coefficient information to estimate the CSI gain is proposed. 
In a third approach, the gain errors are smoothed out by filtering the CSI amplitude. Among them, \cite{Liu2014, Liu2015} use a Hampel filter-based outlier removal and perform interpolation, \cite{Zeng2018, Dou2021} use a Savitzky Golay filter, \cite{Sameera2018} uses a discrete wavelet filter and \cite{Wang2018} uses a median average filter on the CSI amplitude.

\subsubsection{Accounting for phase errors} In a first approach, metrics robust to phase errors are used for sensing. Among them, \cite{Liu2014, Liu2015, Sameera2018} use only CSI amplitude information for sensing. 
Exploiting the fact that phase-errors across RX antennas are correlated, \cite{Wang2017b, Wang2017c} use phase difference, \cite{Zeng2018} uses the conjugate product and \cite{Zeng2019} uses the ratio, respectively, of the CSI at two RX antennas as the sensing metric. In \cite{Chen2018}, a time reversal resonating strength metric is proposed to filter out pair-wise phase differences between any two CSIs. 
In a second approach, the phase errors are explicitly estimated and removed. A popular method is to perform a least-squares fit on the unwrapped CSI phase to estimate out the phase errors \cite{Kotaru2015, Ma2018, Wang2018, Xie2019, Dou2021, Li2022}.

Approach 1 for gain or phase correction does not utilize the full information encoded in CSI, and is thus sub-optimal. It also uses some non-linear metrics of CSI which can only cater to specific sensing applications, and can't replicate the true channel. 
Many of these also have restrictions, such as requiring at least two RX antennas. Approach 3 for gain correction can reduce gain variations but at the cost of distorting the CSI. While approach 2 of gain and phase correction is versatile, can potentially replicate the true channel (thus, allowing use of both the amplitude and phase information), existing solutions are heuristic and do not exploit the structure present in the gain and phase errors. Furthermore, although their performance for some specific sensing tasks has been investigated, a comparative study of how well these methods can replicate the true CSI has not been undertaken.

In this work, we develop a mathematical model for the time-varying true channel, and also for the RX-induced gain error and the synchronization-based phase errors between the TX and RX, which are validated by experiments. Based on these models, we derive multiple theoretically justified methods for approach 2, to estimate and remove the CSI gain and phase errors. Subsequently, we perform a detailed performance comparison of the different gain and phase correction methods for replicating the true CSI and improving the accuracy of an example sensing task. The contributions of this work are as follows:
\begin{itemize}
\item We propose a detailed mathematical WiFi CSI model that takes into consideration RX gain error and synchronization-induced phase error. 
\item Leveraging the distribution of the gain error, we propose several preprocessing algorithms to estimate and remove the gain error in the CSI, including a maximum-likelihood (ML) estimator. 
\item We propose several preprocessing algorithms to remove the phase error present in WiFi CSI, which offer different trade-offs between accuracy and computational complexity. Also among them is a conditionally optimal ML estimator. 
\item We validate the models for the gain error and synchronization-induced phase error using CSI measurements from two real-world WiFi test beds.
\item We perform a detailed performance evaluation of the proposed algorithms for gain and phase error estimation on simulated CSI data under varying scenarios. 
\item Finally, we study the performance of these CSI preprocessing algorithms in a real-world test bed for an example sensing problem, viz., estimating the respiration rate of a stationary subject. 
\end{itemize}

The organization of the paper is as follows: the system model is discussed in Section \ref{sec_chan_model}; the modeling of the CSI gain errors  and the estimation and removal of gain errors are discussed in \ref{sec_gain_correct}; the modeling of the CSI phase errors  and the estimation and removal of phase errors are discussed in \ref{sec_CPE_timing_correct}; the evaluation results on both simulated data and real-world data are provided in Section \ref{sec_eval_results}; and the conclusions and future directions are summarized in Section \ref{sec_conclusions}.

\textbf{Notation:} scalars are represented by light-case letters; sets by light-case calligraphic or double bold letters; and functions by scripted letters. Additionally, ${\rm j} = \sqrt{-1}$, ${a}^{*}$ is the complex conjugate of a complex scalar $a$, $\angle a$ represents the phase angle of a complex scalar $a$, $\lfloor a \rfloor$ is the floor function on a real scalar $a$, $| a |$ is the magnitude of a complex scalar $a$ and $\mathrm{Re}\{a\}$ is the real component of a complex scalar $a$. Additionally, $\mathrm{mod}\{a, b\}$ is the modulus function that gives the remainder after dividing $a$ by $b$, $\mathbb{P}(\cdot)$ represents the probability density function, $\mathrm{Uni}[a,b]$ represents a uniform distribution in the range $[a, b]$, $\mathcal{CN}(\mu, \sigma^2)$ represents a circularly-symmetric complex Gaussian distribution with mean $\mu$ and variance $\sigma^2$. Finally, $\mathbb{Z}$ is the set of integers, $\mathbb{R}$ is the set of real numbers, and $\mathbb{C}$ is the set of complex numbers. 

\section{System model} \label{sec_chan_model}
We consider a system setup that has a stationary, single-antenna WiFi access point as the TX and a stationary, single-antenna WiFi station as the RX, as shown in Fig.~\ref{Fig_sytem_illustrate}.\footnote{The results can easily be extended to multiple antennas by repeating the proposed preprocessing methods for CSI from each TX-RX antenna pair.} To enable sensing of any variations in the environment, the TX transmits a sequence of $P$ WiFi channel state information (CSI) acquisition frames\footnote{Such a frame can be, for example, a null data packets \cite{IEEEWiFi2020}.}, indexed as $\mathcal{P} = \{p \in \mathbb{Z} | 0 \leq p < P\}$, that are transmitted periodically at an interval of $T_{\rm rep}$ seconds each. The header and the payload of these CSI frames are encoded using orthogonal frequency division multiplexing (OFDM) over $K$ sub-carriers indexed as: $\mathcal{K} = \{k \in \mathbb{Z} |0 \leq k \leq K\}$ with a symbol duration of $T_{\rm s}$, and are transmitted at a carrier frequency $f_{\rm c}$. For ease of analysis, we assume that the TX oscillator frequency $f_{\rm c}$ doesn't drift significantly across the $P$ CSI frames. The header of the CSI acquisition frame consists of a legacy short training field (L-STF), a legacy long training field (L-LTF), and one or more long training fields (LTFs) each of size one OFDM symbol \cite{IEEEWiFi2020}, as depicted in Fig.~\ref{Fig_sytem_illustrate}. Each transmitted CSI frame passes through a time-varying channel before reaching the RX. 
\begin{figure} 
\includegraphics[width= 0.5\textwidth]{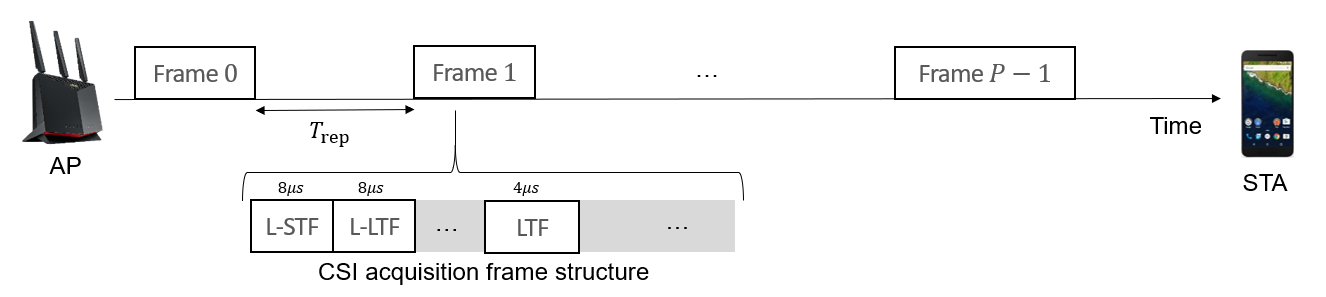}
\caption{An illustration of system model depicting an WiFi access-point, a WiFi station and the CSI acquisition frame structure.}
\label{Fig_sytem_illustrate}
\end{figure}

Let ${h}_{p,k}$ be the frequency-domain channel (that captures the variation to be sensed) between the TX and RX on sub-carrier $k$ for CSI frame $p$. Since the TX and RX are stationary, we model ${h}_{p,k}$ as:
\begin{eqnarray} \label{eqn_channel_resp}
{h}_{p,k} \triangleq {b}_{k} + {d}_{p,k}, 
\end{eqnarray}
where ${b}_{k}$ captures the static component of the channel (i.e., component independent of $p$) corresponding to the line-of-sight path, MPCs that are reflected from walls and static objects, etc., and ${d}_{p,k}$ is the dynamic component of the channel that changes with $p$ and captures the channel variations to be sensed. 
For ease of analysis and to keep the model applicable to a variety of sensing applications, we do not impose any structure on the static component, except that $\angle[ \sum_{k} {b}_{k} {b}_{k+1}^{*}]=0$, without loss of generality. We also define $\sum_k {|{b}_{k}|}^2 \big/ K = \gamma$, where $0 \leq \gamma \leq 1$ is a parameter determining the fraction of channel power in the static component. Similarly, we do not impose any limitation on the Doppler spectrum for the dynamic component ${d}_{p,k}$, and model it as ${d}_{p,k} \sim \mathcal{CN}(0, 1-\gamma)$, with it being independent and identically distributed (i.i.d.) for each $k \in \mathcal{K}$ and $p \in \mathcal{P}$.\footnote{In the simulations, we shall however also consider cases where ${d}_{p,k}$ is sparse in the Doppler domains.} For the theoretical derivations, we shall assume a strongly static channel, i.e., $\gamma \approx 1$ in \eqref{eqn_channel_resp}. 

At the RX, the received signal for CSI frame $p$ is amplified, mixed with a local oscillator signal to convert to base-band, and then sampled using a dedicated analog-to-digital converter (ADC). Note that to ensure the input to the ADC is operating with the right dynamic range, the RX circuit provides a variable gain $g_{p}$. The sampled signal from the ADC is then used to perform symbol start time detection and carrier frequency offset (CFO) compensation before performing OFDM demodulation. For each CSI acquisition frame $p$, the RX estimates the symbol start time using the L-STF \cite{Schmidl1997, Nasir2016}, and we assume that the estimated symbol start time can have an error ${\tau}_p$. Similarly, the RX performs carrier frequency-offset compensation using the L-STF and L-LTF symbols \cite{Sourour2004}. We assume the carrier frequency-offset compensation to be accurate enough to only result in a residual common phase error (CPE) \cite{Petrovic2007, Ratnam2020} ${\psi}_p$. Finally, after the compensation, OFDM demodulation of the frame is performed. The OFDM demodulation output for the LTF symbol is used to estimate the frequency domain channel ${h}_{p,k}$, referred to as channel state information (CSI) on sub-carrier $k$. For any $p \in \mathcal{P}, k \in \mathcal{K}$, this CSI can be modeled as:
\begin{eqnarray} \label{eqn_H_k_est}
\widetilde{{h}}_{p,k} &=& g_{p} {h}_{p,k} e^{-{\rm j} 2 \pi f_k \tau_p} e^{-{\rm j} {\psi}_p}, 
\end{eqnarray}
where $f_k = k/T_{\rm s}$ is the frequency offset of sub-carrier $k$ from the carrier frequency.
Note that, for ease of notation, we do not consider the impact of additive channel noise in \eqref{eqn_H_k_est}. This is without loss of generality because any channel noise can be captured into the dynamic channel component term ${d}_{p,k}$. 
Also note that for $>20$ MHz transmission, the LTF for each $20$MHz segment is rotated by a different phase by the TX as defined by the standard \cite{IEEEWiFi2020}. Here, without loss of generality, we assume that the corresponding de-rotation has been performed by the RX. Also note that non-linearity in the hardware can cause further distortions to \eqref{eqn_H_k_est} \cite{Xie2019}. For ease of analysis, we assume such distortions to be calibrated out apriori.
To help visualize the errors in \eqref{eqn_H_k_est}, we plot the CSI power $\sum_k |\widetilde{h}_{p,k}|/K$ and $\angle \widetilde{h}_{p,k}$ for an example real-world static scenario in Fig.~\ref{Fig_impaired_CSI_illustrate}.
\begin{figure}[!htb]
\centering
\includegraphics[width= 0.24\textwidth]{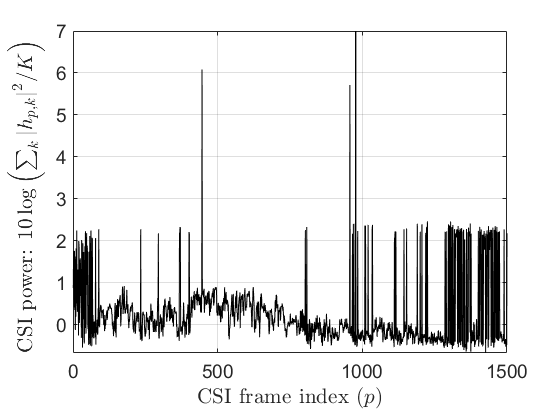}  \hspace{-0.4cm}
\includegraphics[width= 0.25\textwidth]{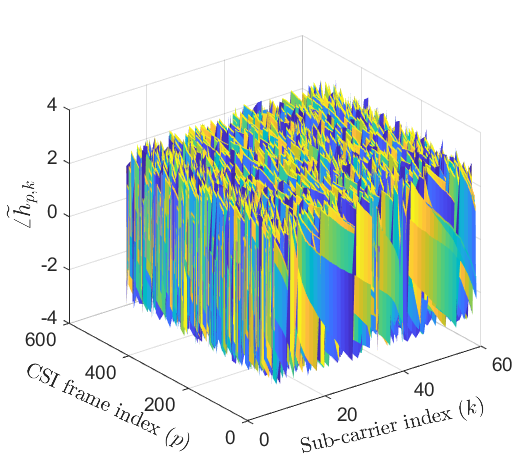}
\caption{A real-world example of CSI power (in decibels) and $\angle \widetilde{h}_{p,k}$ depicting gain and phase errors.}
\label{Fig_impaired_CSI_illustrate}
\end{figure}
In \eqref{eqn_H_k_est}, the gain, timing, and CPE $g_{p}, {\tau}_p, {\psi}_p$ distort the gain and phase of the channel response, which can affect the performance of any sensing algorithm. Therefore, the aim of the paper is to obtain good estimates $\widehat{g}_{p}, \widehat{\tau}_p, \widehat{\psi}_p$ of ${g}_{p}, {\tau}_p, {\psi}_p$, respectively, so that a close approximation to the true channel in \eqref{eqn_channel_resp} can be recovered as:
\begin{eqnarray} \label{eqn_Hk_synchronized}
\widehat{{h}}_{p,k} = \widetilde{{h}}_{p,k} e^{{\rm j} 2 \pi f_k \widehat{\tau}_p} e^{{\rm j} \widehat{\psi}_p} \big/  \widehat{g}_{p}.
\end{eqnarray}
We assume a batch processing framework, where the RX first accumulates the CSI for $P$ CSI acquisition frames $\{\widetilde{{h}}_{p,k} | k \in \mathcal{K}, p \in \mathcal{P}\}$, and then uses them to estimate $\widehat{g}_{p}, \widehat{\tau}_p, \widehat{\psi}_p$ for all $p \in \mathcal{P}$. Key variables of the paper are summarized in Table~\ref{Table_variables}.
\begin{table}[h!]
\centering
\caption{A summary of key variables of the paper.} \label{Table_variables}
\renewcommand{\arraystretch}{1.5} 
\begin{tabular}{| c | l |} 
\hline 
Variable & Description. \\
\hline
$P, K$ & Number of frames, number of OFDM sub-carriers. \\
\hline
$p, k$ & Current frame index, the OFDM sub-carrier index. \\
\hline
$T_{\rm rep}, f_k$ & Inter-frame time, sub-carrier frequency. \\
\hline
$h_{p,k}, \Gamma_p$ & True CSI \eqref{eqn_channel_resp} and its power in dB \eqref{eqn_gamma_true_define}. \\
\hline
$b_{k}, \bar{b}_k$ & CSI static component \eqref{eqn_channel_resp} and its estimate (Sec \ref{sec_CPE_timing_correct}). \\
\hline
$\gamma$ & Fraction of CSI power in static component. \\
\hline
$\widetilde{h}_{p,k}, \widetilde{\Gamma}_p$ & The observed CSI \eqref{eqn_H_k_est} and its power in dB \eqref{eqn_gamma_define}. \\
\hline
$g_p, \widehat{g}_p $ & Gain error \eqref{eqn_H_k_est} and its estimate: \eqref{eqn_baseline_normPower}, \eqref{eqn_baseline_DBscan}, Algos.\ref{Algo5}, \ref{Algo4}. \\
\hline
$g_p^{(1)}, \widehat{g}^{(1)}_p $ & Large-scale gain \eqref{eqn_gain_model} and its estimate \eqref{eqn_proposed_DBscan2}, Algo.\ref{Algo3}. \\
\hline
$g_p^{(2)}, \widehat{g}^{(2)}_p$ & AGC gain \eqref{eqn_gain_model} and its estimate \eqref{eqn_proposed_DBscan}, \eqref{eqn_gamma_hat}. \\
\hline
$\tau_p, \widehat{\tau}_p$ & Timing error \eqref{eqn_H_k_est} and its estimate (see Sec \ref{sec_CPE_timing_correct}). \\
\hline
$\psi_p, \widehat{\psi}_p$ & Common phase error \eqref{eqn_H_k_est} and its estimate (see Sec \ref{sec_CPE_timing_correct}). \\
\hline
$\bar{h}_{p,k}, \widehat{h}_{p,k}$ & The gain-corrected CSI, the final (cleaned) CSI. \\
\hline
\end{tabular}
\end{table}

\section{Estimation of gain errors} \label{sec_gain_correct}
This section aims to obtain an estimate of the gain term $\widehat{g}_{p}$ for each $p \in \mathcal{P}$, and thus generate the gain-corrected CSI:  
\begin{eqnarray} \label{eqn_CSI_gain_comp}
\bar{{h}}_{p,k} \triangleq \widetilde{{h}}_{p,k} \big/ \widehat{g}_{p},
\end{eqnarray}
for each $k \in \mathcal{K}$, for further processing. 

\subsection{Distribution of gain error} \label{subsec_AGC_gain_dist}
When receiving a CSI frame $p$, the receiver gain $g_{p}$ is a combination of the large-scale gain, which adjusts at a slow-time scale, and an automatic gain control (AGC) gain, which updates at a fast-time scale. Let these two gains in decibel (dB) scale be defined as  $g^{(1)}_{p}$ and  $g^{(2)}_{p}$, respectively, i.e., 
\begin{eqnarray} \label{eqn_gain_model}
20 \times \log_{10} \big( g_{p} \big) =  g^{(1)}_{p} + g^{(2)}_{p}
\end{eqnarray}
We assume these gains to follow the following model:
\begin{itemize}
\item The large-scale gain is a slowly varying random process with a Doppler frequency less than $0.1$ Hz and arises due to change in the gain of low-noise amplifier \cite{Lee2007, Cheng2014}, drift in TX oscillator frequency $f_{\rm c}$, etc.
\item The AGC gain is due to the automatic gain control circuit, which for each $p$ picks a value from a discrete grid, i.e., $g^{(2)}_{p} \in \mathcal{G}$ \cite{Jang_VTC2010, ART001948424}. We make the worst case assumption that $g^{(2)}_{p}$ is independently distributed for each $p$ with a uniform prior over elements of $\mathcal{G}$. 
\end{itemize}
Note that the impact of $g^{(1)}_{p}, g^{(2)}_{p}$ is clearly visible in the CSI power in Fig.~\ref{Fig_impaired_CSI_illustrate}. These modeling assumptions are further validated using real-world data in Section \ref{subsec_CSI_gain_phase_err_dist}.  

\subsection{Baseline methods} \label{subsec_gain_baseline_algos}
As a baseline, we consider two methods used in prior art. One popular method \cite{Chen2018, Wang2017b}  is to assume all variation in the CSI power comes from the gain error, i.e.,
\begin{eqnarray} \label{eqn_baseline_normPower}
\widehat{g}_{p} = {\left[\frac{1}{K} \sum_{k \in \mathcal{K}} {|\widetilde{{h}}_{p,k}|}^2 \right]}^{1/2}
\end{eqnarray}
Note that while this gets rid of the gain error, it removes most of the information about the sensing signal encoded in the CSI amplitude. For example, as shown in \cite{Zeng2018}, sometimes all of the sensing signal is captured only inside the CSI power.

As another baseline, we consider an extension of the clustering based approach proposed in \cite{Liu2021}. In this approach for each CSI frame $p \in \mathcal{P}$ we compute the pre-compensated CSI powers as:
\begin{eqnarray} \label{eqn_gamma_define}
\widetilde{\Gamma}_{p} \triangleq 10 \times \log_{10} \left[ \frac{1}{K} \sum_{k \in \mathcal{K}} {|\widetilde{{h}}_{p,k}|}^2 \right]
\end{eqnarray}
We then run the DBSCAN clustering algorithm \cite{DBScan} on $\{\widetilde{\Gamma}_{p} | p \in \mathcal{P}\}$ to cluster the $P$ samples into different AGC gain levels. For each identified cluster, the mean value of $\widetilde{\Gamma}_{p}$ is used to estimate the gain term for that cluster as: 
\begin{eqnarray} \label{eqn_baseline_DBscan}
\widehat{g}_{p} = {10}^{\sum_{q \in \mathcal{P}_p}\widetilde{\Gamma}_{q} \big/\big( 20 |\mathcal{P}_p|\big)},
\end{eqnarray}
where $\mathcal{P}_p$ is set of all frames $q \in \mathcal{P}$ that are assigned to same cluster as $p$ in DBSCAN. Since AGC gain steps are usually larger than $0.2$ dB, we can set parameters $\epsilon = 0.15$ and $\textrm{min-points} = 1$ in DBSCAN. 
One issue with such clustering is that it can perform quite poorly when the large-scale gain has a drift. Additionally, the clustering doesn't exploit the structure present in ${g}^{(2)}_{p}$ thus limiting the performance. Finally, clustering is computationally expensive requiring significant computation time (see Section \ref{subsec_eval_simulated}).

\subsection{Gain estimation for arbitrary $\mathcal{G}$} \label{subsec_gain_arbitraryG}
Here we assume that $\mathcal{G}$ can be an arbitrary but finite, discrete set. Let us define the CSI power for the true channel as:
\begin{eqnarray} \label{eqn_gamma_true_define}
{\Gamma}_{p} \triangleq 10 \times \log_{10} \left[ \frac{1}{K} \sum_{k \in \mathcal{K}} {|{{h}}_{p,k}|}^2 \right].
\end{eqnarray}
We then have the following observation:
\begin{lemma} \label{remark1}
When $\gamma \approx 1$, the true channel gain ${\Gamma}_{p}$ is i.i.d. Gaussian distributed for each $p \in \mathcal{P}$ with a zero mean and a small variance of $\sigma_{\Gamma}^2 = 100 (1 - \gamma^2)\big/K$. 
\end{lemma}
\begin{proof}
See Appendix \ref{appdix1p5}.
\end{proof}
Let us also define the uncompensated channel power increments as:
\begin{eqnarray} \label{eqn_delta_gamma_defn}
\Delta \widetilde{\Gamma}_{p} \triangleq \widetilde{\Gamma}_{p} - \widetilde{\Gamma}_{p-1},
\end{eqnarray}
where we use $\widetilde{\Gamma}_{-1}=\widetilde{\Gamma}_{0}$. Using \eqref{eqn_H_k_est}, \eqref{eqn_gain_model}, \eqref{eqn_gamma_define} and \eqref{eqn_gamma_true_define}, we have: 
\begin{eqnarray}
\widetilde{\Gamma}_{p} &=& {\Gamma}_{p} + g^{(1)}_{p} + g^{(2)}_{p} \label{eqn_tilde_gamma_approx} \\
\Rightarrow \Delta \widetilde{\Gamma}_{p} & \approx & {\Gamma}_{p} - {\Gamma}_{p-1} + g^{(2)}_{p} - g^{(2)}_{p-1} , \nonumber
\end{eqnarray}
which is independent of the large-scale gain $g^{(1)}_{p}$ due to its low bandwidth. Using Lemma \ref{remark1}, we can then infer that $\Delta \widetilde{\Gamma}_{p}$ is an estimate of $g^{(2)}_{p} - g^{(2)}_{p-1}$ with an additive zero-mean Gaussian noise. Since $g^{(2)}_{p}$ comes from a finite discrete set $\mathcal{G}$, as a first step we run DBSCAN algorithm on $\{ \Delta \widetilde{\Gamma}_{p} | 1 \leq p < P \}$ to cluster AGC gain increments. Note that unlike \eqref{eqn_baseline_DBscan}, the clustering here is resilient to the large-scale gain variations. For each $p \in \mathcal{P}$ we can then estimate the AGC gain for frame $p$ sequentially as: 
\begin{eqnarray} \label{eqn_proposed_DBscan}
\widehat{g}^{(2)}_{p} = \widehat{g}^{(2)}_{p-1} + \sum_{q \in \mathcal{P}_p} \Delta \widetilde{\Gamma}_{q} \big/ |\mathcal{P}_p|,
\end{eqnarray}
where we use $\widehat{g}^{(2)}_{-1}=0$ and $\mathcal{P}_p$ is set of all frames $q \in \mathcal{P}$ that are assigned to same cluster as $p$ in DBSCAN. The large-scale gain (along with any accumulated error in \eqref{eqn_proposed_DBscan}) is then estimated using a low-pass filter as:
\begin{eqnarray} \label{eqn_proposed_DBscan2}
\widehat{g}^{(1)}_{p} = \textrm{LPF}_{\text{0.1 Hz}}\{ \widetilde{\Gamma}_{p} - \widehat{g}^{(2)}_{p}\},
\end{eqnarray}
where $\textrm{LPF}_{\text{0.1 Hz}}$ is a low-pass filter with cut-off frequency $0.1$ Hz. Finally $\widehat{g}_{p}$ is computed, as summarized in Algorithm \ref{Algo5}. 
\begin{algorithm}
\label{Algo5}
\caption{Gain estimation with arbitrary $\mathcal{G}$}
\begin{algorithmic} 
\STATE Given: $\widetilde{\Gamma}_{p}$ for $p \in \mathcal{P}$
\STATE Compute $\Delta \widetilde{\Gamma}_{p} = \widetilde{\Gamma}_{p} - \widetilde{\Gamma}_{p-1}$ for each $p \in \mathcal{P}$.
\STATE Run DBSCAN on $\{ \Delta \widetilde{\Gamma}_{p} | 1 \leq p < P \}$ with $\epsilon = 0.2$ and $\textrm{min-points} = 1$.
\FOR{$p \in \mathcal{P}$}
\STATE Compute $\widehat{g}^{(2)}_{p}$ from \eqref{eqn_proposed_DBscan}.
\ENDFOR
\STATE Compute $\widehat{g}^{(1)}_{p}$ from \eqref{eqn_proposed_DBscan2} for all $p \in \mathcal{P}$.
\STATE //We use a moving average filter for $\mathrm{LPF}_{\text{0.1 Hz}}$ with one-sided width $6/T_{\rm rep}$.
\STATE Return $\big\{\widehat{g}_{p} = {10}^{(\widehat{g}^{(1)}_{p}+\widehat{g}^{(2)}_{p})/20} \big| p \in \mathcal{P} \big\}$.
\end{algorithmic}
\end{algorithm}

\subsection{Gain estimation for uniformly-spaced $\mathcal{G}$} \label{subsec_gain_uniformG}
In this section, we assume a specific structure on the AGC gain set as: $\mathcal{G} = \{z \lambda | z \in \mathbb{Z}\}$, i.e., the AGC gain ${g}^{(2)}_{p}$ is always in multiples of an unknown constant $\lambda$. This causes the following observation:
\begin{lemma} \label{remark2}
All information in $\widetilde{\Gamma}_{p}$ about ${\Gamma}_{p}$ is captured within: $\mathrm{mod}\{ \widetilde{\Gamma}_{p}, \lambda \}$, and equivalently, within 
\begin{eqnarray} \label{eqn_defn_Xi_tilde}
\widetilde{\Xi}_{p} \triangleq \exp\left\{ 2 \pi {\rm j} \widetilde{\Gamma}_{p} \big/ \lambda \right\}.
\end{eqnarray}
\end{lemma}
\begin{proof}
See Appendix \ref{appdix1p5}.
\end{proof}
For convenience, let us assume that the value of the AGC gain step size $\lambda$ is known. A search over different hypothesis of $\lambda$ is performed, as explained later in the section. For a given value of $\lambda$, the proposed method first finds the best estimate of large-scale gain $g^{(1)}_{p}$ and then estimates the gains $\{ \widehat{g}_{p} | p \in \mathcal{P} \}$.

\subsubsection{Estimation of large-scale gain $g^{(1)}_{p}$} 
Using Lemma \ref{remark2}, note that we can expand:
\begin{eqnarray}
\widetilde{\Xi}_{p} = \exp\left\{ {\rm j} 2 \pi \big[{\Gamma}_{p} + g^{(1)}_{p} \big] \big/ \lambda \right\}
\end{eqnarray}
Since ${\Gamma}_{p}$ is zero-mean Gaussian distributed with a small variance (from Lemma \ref{remark1}), it follows that $\Xi_{p}$ undergoes a slow ($\leq 0.1$ Hz) rotation in the complex plane with $p$ due to $g^{(1)}_{p}$. Correspondingly, $g^{(1)}_{p}$ can be estimated by low-pass filtering as shown in Algorithm \ref{Algo3}. For the rest of the section we shall assume $\widehat{g}^{(1)}_{p} \approx g^{(1)}_{p}$, for convenience.
\begin{algorithm}
\caption{Estimation of $g^{(1)}_{p}$} \label{Algo3}
\begin{algorithmic} 
\STATE Given: $\widetilde{\Xi}_{p}$ for $p \in \mathcal{P}$.
\STATE Compute $\bar{\Xi}_{p} = \mathrm{LPF}_{\text{0.1 Hz}}\{ \widetilde{\Xi}_{p} \}$.
\STATE //We use a moving average filter for $\mathrm{LPF}_{\text{0.1 Hz}}$ with one-sided width $6/T_{\rm rep}$.
\STATE Return $\widehat{g}^{(1)}_{p} = \mathscr{U} \big[ \angle \bar{\Xi}_{p} \big] \lambda / 2\pi$ for $p \in \mathcal{P}$.
\STATE // Here $\mathscr{U}(\cdot)$ is the phase unwrapping function that for each $p \in \mathcal{P}$ adds integer shifts of $2 \pi$ to an argument to ensure that $|\angle\bar{\Xi}_{p} - \angle\bar{\Xi}_{p-1}| \leq \pi$.
\end{algorithmic}
\end{algorithm}

\subsubsection{Estimation of $\widehat{g}_{p}$ for each $p$} 
Given estimate $\widehat{g}^{(1)}_{p}$ from Algorithm \ref{Algo3} is accurate, we obtain from \eqref{eqn_tilde_gamma_approx}:
\begin{eqnarray}
\widetilde{\Gamma}_{p} - \widehat{g}^{(1)}_{p} \approx \Gamma_{p} + {g}^{(2)}_{p}. \nonumber 
\end{eqnarray}
Since $\Gamma_{p}$ is zero-mean Gaussian distributed, and ${g}^{(2)}_{p} = z_p \lambda$ for some $z_p \in \mathbb{Z}$, it can be readily shown that the ML estimate of ${g}^{(2)}_{p}$ is: 
\begin{eqnarray} \label{eqn_gamma_hat}
\widehat{g}^{(2)}_{p} = \lambda \bigg( \bigg\lfloor \frac{ \widetilde{\Gamma}_{p} - \widehat{g}^{(1)}_{p} }{\lambda} + \frac{1}{2} \bigg\rfloor - \frac{1}{2} \bigg). 
\end{eqnarray}
and therefore the ML estimate of the CSI power is: 
\begin{eqnarray} \label{eqn_hat_Gama}
\widehat{\Gamma}_{p} = \widetilde{\Gamma}_{p} - \widehat{g}^{(1)}_{p} - \widehat{g}^{(2)}_{p} .
\end{eqnarray}
%
\begin{figure} 
\includegraphics[width= 0.5\textwidth]{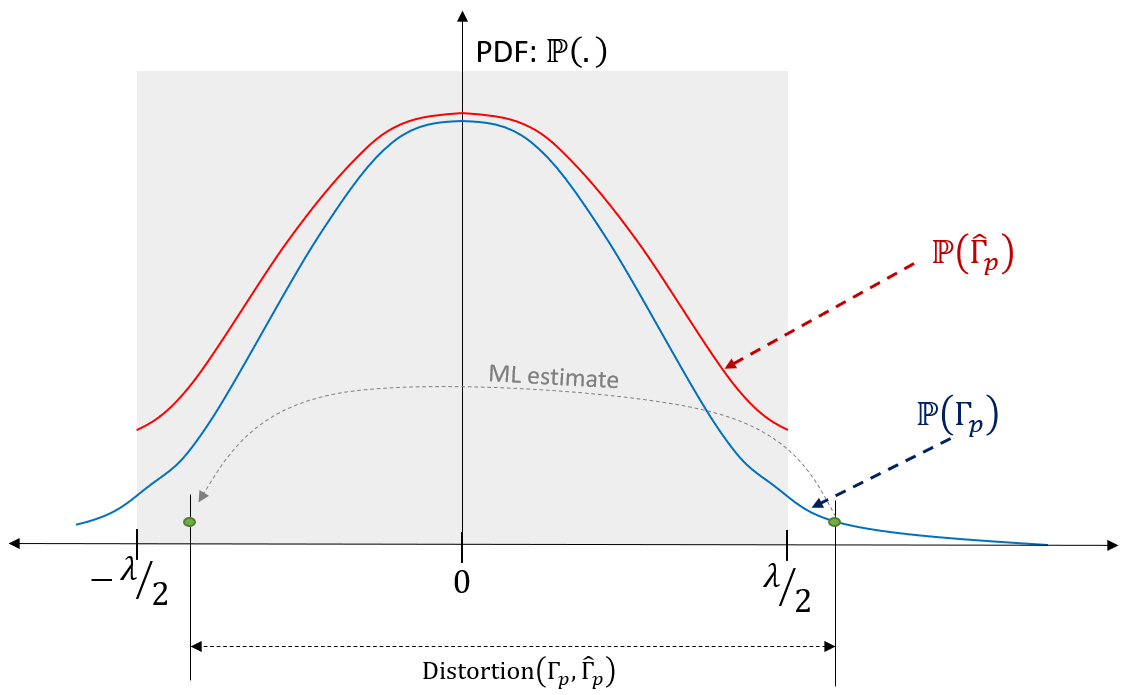}
\caption{An illustration of ${\Gamma}_{p}$, its ML estimate $\widehat{\Gamma}_{p}$, and the distortion introduced by the ML estimation.}
\label{Fig_gamma_distortion}
\end{figure}
The distribution of $\widehat{\Gamma}_{p}$ and its relation to ${\Gamma}_{p}$ is illustrated in Fig.~\ref{Fig_gamma_distortion}. We have the following observations on this ML estimate $\widehat{\Gamma}_{p}$:
\begin{lemma} \label{remark3}
If estimate $\widehat{g}^{(1)}_p$ is accurate and $P \gg 1$, the then ${\Gamma}_{p}$ has a wrapped Gaussian distribution and its variance can be estimated from $\widetilde{\Gamma}_{p}$ as:
\begin{eqnarray} \label{eqn_sigma_est}
\widehat{\sigma}_{\Gamma}^2 = - \frac{\lambda^2}{2 \pi^2} \log \bigg| \frac{1}{P} \sum_{p \in \mathcal{P}}  \exp \Big\{ {\rm j} 2 \pi \widehat{\Gamma}_{p} / \lambda \Big\} \bigg|
\end{eqnarray}
\end{lemma}
\begin{proof}
See Appendix \ref{appdix1p5}
\end{proof}
\begin{lemma} \label{remark4}
If estimate $\widehat{g}^{(1)}_p$ is accurate, the distortion between ${\Gamma}_{p}$ and $\widehat{\Gamma}_{p}$ is given by ${\lambda}^2 \mathscr{D} \big( \lambda \big/ \sigma_{\Gamma} \big)$ where:
\begin{eqnarray} \label{eqn_distortion}
\mathscr{D}( x ) = \sum_{z \in \mathbb{Z}} \bigg[ Q \bigg( \Big(z-\frac{1}{2}\Big) x \bigg) - Q \bigg( \Big(z+\frac{1}{2} \Big) x \bigg) \bigg] {z}^2
\end{eqnarray}
where $Q(\cdot)$ is the Q-function for standard Gaussian.
\end{lemma}
\begin{proof}
See Appendix \ref{appdix1p5}.
\end{proof}

\subsubsection{Estimation of $\lambda$}
The last step of the algorithm is to select the best hypothesis of $\lambda$ if its value is unknown apriori. For this, we can perform a line-search over a feasible set, where the goal is to find a candidate $\lambda$ that (i) minimizes the model fitting error, captured by $\widehat{\sigma}_{\Gamma}^2$ in Lemma \ref{remark3} and (ii) keeps the distortion caused by the model low, captured by Lemma \ref{remark4}. Correspondingly we can select $\lambda$ that minimizes:
$$ \widehat{\sigma}_{\Gamma}^2 + {\lambda}^2 \mathscr{D}\big( \lambda \big/ \widehat{\sigma}_{\Gamma} \big) . $$
\begin{remark}
Intuitively, $\widehat{\sigma}_{\Gamma}^2$ captures the error in `model fitting' which will be smaller for a smaller $\lambda$ (due to more AGC gain levels). To counter act that, the distortion terms acts like a regularizer that incentivizes larger choices of $\lambda$. 
\end{remark}
The overall set of steps in the proposed algorithm are summarized in Algorithm \ref{Algo4}.
\begin{algorithm}
\label{Algo4}
\caption{Gain estimation with uniform $\mathcal{G}$}
\begin{algorithmic} 
\STATE Given: $\widetilde{\Gamma}_{p}$ for $p \in \mathcal{P}$
\STATE $\lambda_{\rm max} = 1.5 \big( \max_p\{\widetilde{\Gamma}_{p}\} - \min_p\{\widetilde{\Gamma}_{p}\} \big)$
\FOR{$\lambda = (0.05:0.05:1)\times \lambda_{\rm max}$}
\STATE // Can be computed in parallel to improve speed.
\STATE Compute $\widetilde{\Xi}_{p}$ from \eqref{eqn_defn_Xi_tilde} for $p \in \mathcal{P}$.
\STATE Run Algorithm \ref{Algo3} to obtain $\widehat{g}^{(1)}_{p}$.
\STATE Compute $\widehat{g}^{(2)}_p, \widehat{\Gamma}_{p}$ from \eqref{eqn_gamma_hat} and \eqref{eqn_hat_Gama} for $p \in \mathcal{P}$.
\IF{$\sum_p {(\widehat{\Gamma}_{p})}^2 \big/ P \leq \lambda^2/24$}
\STATE Compute $\widehat{\sigma}_{\Gamma}$ from \eqref{eqn_sigma_est}.
\STATE Compute $\mathscr{D} \big( \lambda \big/ \widehat{\sigma}_{\Gamma} \big)$ from \eqref{eqn_distortion}.
\STATE // $\mathscr{D}(\cdot)$ can be pre-computed and stored in a look up table to improve speed.
\STATE Compute $\mathrm{Obj} (\lambda) = \widehat{\sigma}_{\Gamma}^2 + {\lambda}^2 \mathscr{D} \big( \lambda \big/ \widehat{\sigma}_{\Gamma} \big)$.
\ELSE
\STATE $\mathrm{Obj} (\lambda_{m}) = \infty$.
\STATE // Skip $\lambda$ where variance is too close to uniform distribution. Eqn. \eqref{eqn_sigma_est} is inaccurate in such cases.
\ENDIF
\ENDFOR
\STATE For $\lambda$ with minimum value of $\mathrm{Obj}(\lambda)$, return $\big\{\widehat{g}_{p} = {10}^{(\widehat{g}^{(1)}_p+\widehat{g}^{(2)}_p)/20} \big| p \in \mathcal{P} \big\}$.
\end{algorithmic}
\end{algorithm}

\section{Estimation of timing error and CPE} \label{sec_CPE_timing_correct}
In this section, the gain corrected CSI $\bar{{h}}_{p,k}$ shall be used for estimating the values of ${\tau}_p, {\psi}_p$. For the analysis in this section, we assume that these gain estimates are accurate, and thus from \eqref{eqn_H_k_est} we have:
\begin{eqnarray} \label{eqn_Hk_gain_synchronized}
\bar{{h}}_{p,k} = {{h}}_{p,k} e^{-{\rm j} 2 \pi f_k {\tau}_p} e^{-{\rm j} {\psi}_p}.
\end{eqnarray}

\subsection{Distribution of timing error and CPE} \label{subsec_timing_CFO_dist}
The timing error ${\tau}_{p}$ can arise from two sources: (i) the sampling time granularity of the ADC at the RX and (ii) the error in the symbol start time estimated using the L-STF of the CSI frame, and it is fundamentally limited by the inverse of the OFDM system bandwidth, viz., $T_{\rm s}/K$. Therefore it is safe to model the timing error as $\tau_{p} \sim \textrm{Uni}(- \kappa T_{\rm s}/K, \kappa T_{\rm s}/K)$, where $\kappa$ is a system parameter that depends on the accuracy of the RX timing compensation (typically $\kappa < 20$). Furthermore, since synchronization is performed independently for each CSI frame, the timing error ${\tau}_{p}$ is also independently distributed for each $p \in \mathcal{P}$.

The CPE ${\psi}_p$ is the difference in the carrier phase between the TX and RX at time $p T_{\rm rep}$. This suggests that ${\psi}_p \approx {\rm mod} \left\{ {\psi}_{1} + 2 \pi f_{\rm CFO} (p-1)T_{\rm rep}, 2\pi \right\} $, where $f_{\rm CFO}$ is the CFO. However, the value of $f_{\rm CFO}$ is typically unknown, it has a short coherence time, and is hard to estimate with the necessary precision to exploit this relation. In addition, cycle slips in phase-locked loop circuits present in the TX and RX further cause phase decoherence and break down this relationship. 
Correspondingly, ${\psi}_p$ is modeled to be independent and uniformly distributed for each $p \in \mathcal{P}$, i.e., $\psi_p \sim \mathrm{Uni}[-\pi, \pi]$. Note that this also explains the experimental observations made in \cite{Liu2014} about the (un-corrected) channel phase, and the behavior of $\angle \widetilde{h}_{p,k}$ in Fig.~\ref{Fig_impaired_CSI_illustrate}. Further experimental validation of these models is performed later in Section \ref{subsec_CSI_gain_phase_err_dist}. 

\subsection{Baseline methods}
As a baseline, we consider two methods from prior art. First, is the heuristic method for the correction of the timing error and CPE $\widehat{\tau_p}, \widehat{\psi}_p$ that was proposed in \cite{Kotaru2015, Dou2021}:
\begin{align}
\widehat{\tau}_p, \widehat{\psi}_{p} = \argmin_{\tau, \psi} \bigg\{ \sum_{k \in \mathcal{K}} {\Big( 2\pi f_{k} \tau + \psi - \mathscr{U}\{\angle \bar{{h}}_{p,k} \} \Big)}^2 \bigg\} ,  \label{eqn_CPE_timing_baseline1} 
\end{align}
where $\mathscr{U}(\cdot)$ is the phase unwrapping function that for each $k \in \mathcal{K}$ adds integer shifts of $2 \pi$ to an argument to ensure that $|\angle \bar{{h}}_{p,k} - \angle \bar{{h}}_{p,k-1}| \leq \pi$. Note that this method does not exploit the temporal correlation of the channel. In fact, it can be shown to be an approximation of one of our proposed solutions below. 
A second heuristic method to estimate ${\tau}_p$ was proposed in the IEEE 802.11az standard \cite{IEEE_11az}, that exploits the frequency domain, coherence of the channel as:
\begin{subequations} \label{eqn_CPE_timing_baseline2} 
\begin{eqnarray}
\widehat{\tau}_p = \frac{T_{\rm s}}{2 \pi} \angle \bigg[ \sum_{k=0}^{K-2} \bar{{h}}_{p,k} \bar{{h}}_{p,k+1}^{*} \bigg], \label{eqn_CPE_baseline2_tau} 
\end{eqnarray}
where $T_{\rm s}$ is the OFDM symbol duration. This method can be extended to obtain the phase estimate as:
\begin{eqnarray}
\widehat{\psi}_{p} = - \angle \bigg[ \sum_{k \in \mathcal{K}} \bar{h}_{p,k} e^{{\rm j} 2 \pi f_k \widehat{\tau}_p}\bigg] . \label{eqn_CPE_baseline2_psi}
\end{eqnarray}
\end{subequations}
It can be shown that \eqref{eqn_CPE_baseline2_tau} is reasonably accurate when $\gamma \approx 1$ and \eqref{eqn_CPE_baseline2_psi} is accurate when additionally the channel has a strong LoS path. However, this method also does not exploit the temporal correlation of the channel and hence is sub-optimal.

\subsection{Method 1: Estimation in strongly LoS channels} \label{subsec_timing_CFO_est_0}
Here we consider estimation of $\tau_p, \psi_p$ under the assumption that the channel has a strong LoS component, where $\gamma \approx 1$ and the static channel component is dominated by a strong frequency-flat path. 
Under these conditions, using \emph{coarse} estimates $\bar{\tau}_p, \bar{\psi}_p$ from \eqref{eqn_CPE_timing_baseline2}, an estimate of the static channel component can be obtained as:
\begin{eqnarray} \label{eqn_b_hat_est_LoS}
\bar{b}_k &=& \sum_{p\in \mathcal{P}} \bar{h}_{p,k} e^{{\rm j}(2 \pi f_k \bar{\tau}_p + \bar{\psi}_p)} \big/ P 
\end{eqnarray}
Note that $\bar{b}_k$ in \eqref{eqn_b_hat_est_LoS} is averaged over the $P$ CSI frames and hence is expected to be less noisy than estimates $\bar{\tau}_p, \bar{\psi}_p$. 
Assuming the estimate $\bar{b}_k$ is accurate for all $k \in \mathcal{K}$, from \eqref{eqn_Hk_gain_synchronized} the ML estimate of $\{{\tau}_p, {\psi}_p | p \in \mathcal{P} \}$ can further be obtained as:
\begin{eqnarray} \label{eqn_timing_err_est_method1}
\widehat{\tau}_p, \widehat{\psi}_p = \argmin_{|\tau | \leq \frac{\kappa T_{\rm s}}{K}, |\psi| < \pi} \bigg\{ \sum_{k \in \mathcal{K}} {\Big|\bar{{h}}_{p,k} e^{{\rm j} (2 \pi f_k \tau + \psi)} - \bar{b}_k \Big|}^2 \bigg\} \nonumber \\
\equiv \argmax_{|\tau | \leq \frac{\kappa T_{\rm s}}{K}, |\psi| < \pi} \bigg\{ \sum_{k \in \mathcal{K}} \mathrm{Re} \Big\{ e^{-{\rm j} (2 \pi f_k \tau + \psi)} {\bar{{h}}_{p,k}}^{*} \bar{{b}}_k \Big\} \bigg\}.
\end{eqnarray}
To further simplify the complexity of the line-search in \eqref{eqn_timing_err_est_method1}, we also have the following result:
\begin{lemma} \label{Th_method1_estimation_simpl}
In the strongly LoS scenario where ${b}_{k} = \bar{b}_k$ and $\gamma \approx 1$, the solutions to \eqref{eqn_timing_err_est_method1} are also the solutions of the weighted least-squares problem:
\begin{flalign}
& \widehat{\tau}_p, \widehat{\psi}_p = \argmin_{\tau, \psi} \bigg\{ \sum_{k \in \bar{\mathcal{K}}} \big|\omega_{p,k} \big| & \nonumber \\
& \qquad \qquad \qquad \Big[ 2\pi f_{k} (\tau-\bar{\tau}_p) + \psi - \mathscr{U} \big( \angle \omega_{p,k} \big) \Big]^2 \bigg\}, \!\!\!\!\! & \label{eqn_CPE_timing_simpl_method1} 
\end{flalign}
where $\omega_{p,k} = \bar{{{h}}}_{p,k}^{*} \bar{b}_{k} e^{-{\rm j} 2 \pi f_k \bar{\tau}_p} $, $\bar{\mathcal{K}} \triangleq \{k \in \mathcal{K} \big| |\bar{b}_k| > 0\}$, $\bar{\tau}_p$ is the estimate from \eqref{eqn_CPE_baseline2_tau}, $\widetilde{\mathscr{U}}(\cdot)$ is a `robust' phase unwrapping function computed as: 
\begin{flalign}
& \widetilde{\mathscr{U}} \big(\angle \omega_{p,k} \big) = \mathrm{mod}\bigg\{ \angle \omega_{p,k} - \mathscr{U} \bigg( \angle \bigg[\sum_{\ell \in \mathcal{L}(k)} \omega_{p,\ell} \bigg] \bigg) + \pi, 2 \pi \bigg\} & \nonumber \\
& \qquad \qquad \qquad  - \pi + \mathscr{U} \bigg( \angle \bigg[ \sum_{\ell \in \mathcal{L}(k)} \omega_{p,\ell} \bigg] \bigg), &
\end{flalign}
and $\mathcal{L}(k)$ is a set that includes $k$, its $3$ preceding, and $3$ succeeding indices in sorted set $\bar{\mathcal{K}}$. 
\end{lemma}
\begin{proof}
See Appendix \ref{appdix0}.
\end{proof}
Here $\widetilde{\mathscr{U}}(\cdot)$ is a \emph{robust} alternative to $\mathscr{U}(\cdot)$, that reduces impact of noise before unwrapping \cite{Xu2013}. Additionally, the pre-computation of $\bar{\tau}_p$ and its use in $\omega_{p,k}$, although not essential, is beneficial to prevent errors in phase unwrapping when there is a long discontinuity in $\bar{\mathcal{K}}$ owing to deep fades. 
Note that the solution to the weighted least-squares problem can be found in closed-form, unlike \eqref{eqn_timing_err_est_method1}, thus reducing complexity. Using \eqref{eqn_timing_err_est_method1} or \eqref{eqn_CPE_timing_simpl_method1}, we propose a CPE and timing offset estimation approach for each frame $p \in \mathcal{P}$, as shown in Algorithm \ref{Algo0}. 
\begin{remark}
Note that the baseline method \eqref{eqn_CPE_timing_baseline1} is an approximation of \eqref{eqn_CPE_timing_simpl_method1} obtained by setting $|{b}_k|= 1$ and $\bar{\tau}_p = 0$ and using conventional phase-unwrapping.
\end{remark}
\begin{algorithm}
\label{Algo0}
\caption{Strong LoS estimation of $\widehat{\tau}_p, \widehat{\psi}_p$}
\begin{algorithmic} 
\STATE Given: $\bar{{{h}}}_{p,k}$ for $p \in \mathcal{P}, k \in \mathcal{K}$
\FOR{$p \in \mathcal{P}$}
\STATE Compute $\bar{\tau}_p, \bar{\psi}_p$ from \eqref{eqn_CPE_timing_baseline2}.
\STATE{// Can be computed in parallel to improve speed.}
\ENDFOR
\FOR{$k \in \mathcal{K}$}
\STATE Compute $\bar{b}_k$ from \eqref{eqn_b_hat_est_LoS}.
\STATE{// Can be computed in parallel to improve speed.}
\ENDFOR
\FOR{$p \in \mathcal{P}$}
\STATE Calculate and save $\widehat{\tau}_p, \widehat{\psi}_p$ from \eqref{eqn_timing_err_est_method1} or \eqref{eqn_CPE_timing_simpl_method1}.
\STATE{// Can be computed in parallel to improve speed.}
\STATE{// For \eqref{eqn_CPE_timing_simpl_method1} we use $\bar{\mathcal{K}} = \{k \in \mathcal{K} {\big| |\bar{b}_k|}^2 > 0.1 \}$}
\ENDFOR 

\end{algorithmic}
\end{algorithm}

\subsection{Method 2: Estimation in strongly static channels} \label{subsec_timing_CFO_est}
In this method, we do not assume the channel response to be frequency flat, although we still assume $\gamma \approx 1$. Unfortunately, the joint ML estimation of $\{{\tau}_p, {\psi}_p | p \in \mathcal{P} \}$ is too computationally cumbersome. Therefore, we shall consider a greedy, causal estimation approach where ${\tau}_p, {\psi}_p$ for each $p \in \mathcal{P}$ are estimated using prior estimates of ${\psi}_q, \tau_q$ for all the past CSI frames $q \in \mathcal{Q}_p$ where $\mathcal{Q}_p \triangleq \{q \in \mathcal{P} | q < p\}$. For convenience, let us define $\widehat{\tau}_q, \widehat{\psi}_q$ as the estimates of ${\tau}_q, {\psi}_q$, respectively, for each $q \in \mathcal{Q}_p$. We then have the following result:
\begin{lemma} \label{Th_ML_estimation}
For any $p \in \mathcal{P}$, if the estimates $\widehat{\tau}_q, \widehat{\psi}_q$ are error-free for $q \in \mathcal{Q}_p$, i.e., if $\widehat{\tau}_q = {\tau}_q, \widehat{\psi}_q = {\psi}_q$, then the conditional ML estimates of $\tau_p, {\psi}_p$ are the solutions of:
\begin{subequations} \label{eqn_synch_error_est}
\begin{eqnarray}
\widehat{\tau_p} &=& \argmax_{|\tau| \leq \frac{\kappa T_{\rm s}}{K}} \bigg| \sum_{k \in \mathcal{K}} \sum_{q \in \mathcal{Q}_p} e^{-{\rm j} 2 \pi f_k \tau} {\bar{{{h}}}_{p,k}}^{*} \widehat{{{h}}}_{q,k} \bigg| \label{eqn_timing_err_est} \\
\widehat{\psi}_p &=& {\rm angle} \bigg\{ \sum_{k \in \mathcal{K}} \sum_{q \in \mathcal{Q}_p} e^{-{\rm j} 2 \pi f_k \widehat{\tau_p}} {\bar{{{h}}}_{p,k}}^{*} \widehat{{{h}}}_{q,k} \bigg\}, \label{eqn_CPE_est}
\end{eqnarray}
where $\widehat{{{h}}}_{q,k}$ is as defined in \eqref{eqn_Hk_synchronized}.
\end{subequations}
\end{lemma}
\begin{proof}
See Appendix \ref{appdix0}
\end{proof}
To further simplify the complexity of the line-search in \eqref{eqn_timing_err_est}, we also have the following lemma:
\begin{lemma} \label{Th_ML_estimation_simpl}
In the quasi-static channel scenario where $\gamma \approx 1$ if for any $p \in \mathcal{P}$ the conditions of Lemma \ref{Th_ML_estimation} are satisfied, then the conditional ML estimates of $\tau_p, {\psi}_p$ are the solutions of the weighted least-squares problem:
\begin{flalign}
& \widehat{\tau_p}, \widehat{\psi}_p = \argmin_{\tau, \psi} \bigg\{ \sum_{k \in \bar{\mathcal{K}}} |\omega_{p,k}| & \nonumber \\
& \qquad \qquad \qquad {\Big[ \big(2\pi f_{k} (\tau - \bar{\tau}_p) + \psi - \widetilde{\mathscr{U}} \big( \angle \omega_{p,k} \big) \Big]}^2 \bigg\}, \!\!\!\! & \label{eqn_CPE_timing_simpl} 
\end{flalign}
where $\omega_{p,k} = {\bar{{{h}}}_{p,k}}^{*} \big( \sum_{q \in \mathcal{Q}_p} \widehat{{{h}}}_{q,k} \big) e^{-{\rm j} 2 \pi f_{k} \bar{\tau}_p}$, $\widehat{{{h}}}_{q,k}$ is as defined in \eqref{eqn_Hk_synchronized} and $\bar{\mathcal{K}}, \bar{\tau}_p, \widetilde{\mathscr{U}}(\cdot)$ are as defined in Lemma \ref{Th_method1_estimation_simpl}.
\end{lemma}
\begin{proof}
Proof is similar to Lemma \ref{Th_method1_estimation_simpl}.
\end{proof}
Note that the solution to the weighted least-squares problem can be found in closed-form, thus avoiding the search complexity of \eqref{eqn_timing_err_est}. 
Using \eqref{eqn_synch_error_est} or \eqref{eqn_CPE_timing_simpl}, we propose a CPE and timing offset estimation approach for each CSI frame $p \in \mathcal{P}$, as shown in Algorithm \ref{Algo1}. Since the number of past samples $q < p$ is small for small values of $p$, we use the approach in Algorithm \ref{Algo0} for $p < P/10$ to minimize error accumulation. 
\begin{remark}
Intuitively, the algorithm works sequentially such that when dealing with the estimation for frame $p$, CPE and timing errors for all the chronologically preceding CSI frames have already been estimated. The algorithm then searches over the possible phase and timing error corrections for frame $p$ to find the one that ensures the highest correlation of the corrected CSI $\widehat{{{h}}}_{p,k}$ to its preceding corrected CSIs $\widehat{{{h}}}_{q,k}$.
\end{remark}
\begin{algorithm}
\label{Algo1}
\caption{Forward-pass for $\widehat{\tau}_p, \widehat{\psi}_p$}
\begin{algorithmic} 
\STATE Given: $\bar{{{h}}}_{p,k}$ for $p \in \mathcal{P}, k \in \mathcal{K}$
\FOR{$p \in \mathcal{P}$}
\STATE Compute $\bar{\tau}_p, \bar{\psi}_p$ from \eqref{eqn_CPE_timing_baseline2}.
\STATE{// Can be computed in parallel to improve speed.}
\ENDFOR
\FOR{$k \in \mathcal{K}$}
\STATE Compute $\bar{b}_k$ from \eqref{eqn_b_hat_est_LoS}.
\STATE{// Can be computed in parallel to improve speed.}
\ENDFOR
\FOR{$p=0:1:\lfloor P/10 \rfloor$}
\STATE Calculate and save $\widehat{\tau}_p, \widehat{\psi}_p$ from  \eqref{eqn_CPE_timing_simpl_method1}.
\STATE Calculate $\widehat{{{h}}}_{p,k}$ from \eqref{eqn_Hk_synchronized}.
\ENDFOR
\FOR{$p=(\lfloor P/10 \rfloor+1):1:(P-1)$}
\STATE Define $\mathcal{Q}_p = \big\{ q \in \mathcal{P} \big| q < p \big\}$
\STATE Calculate and save $\widehat{\tau}_p, \widehat{\psi}_p$ from \eqref{eqn_synch_error_est} or \eqref{eqn_CPE_timing_simpl}.
\STATE Calculate $\widehat{{{h}}}_{p,k}$ from \eqref{eqn_Hk_synchronized}.
\ENDFOR 
\STATE{// For \eqref{eqn_CPE_timing_simpl_method1} and \eqref{eqn_CPE_timing_simpl} we use $\bar{\mathcal{K}} = \{k \in \mathcal{K} \big| {|\bar{b}_k|}^2 > 0.1 \}$}
\end{algorithmic}
\end{algorithm}
Note that one lacuna of Algorithm \ref{Algo1} is that the estimation at CSI frame $p$ in \eqref{eqn_synch_error_est} only relies on the preceding CSI frames $q < p$. This can cause large errors for the first few CSI frames since there isn't sufficient history, and thus we had to resort to Algorithm \ref{Algo0} for $p < P/10$. To further improve the estimates obtained with Algorithm \ref{Algo1}, an (optional)  backward-pass can also be used as shown in Algorithm \ref{Algo2}, where we redefine $\mathcal{Q}_p = \{q \in \mathcal{P} | q > p\}$ to use the future CSI frames to estimate the CPE and timing offset of the current CSI frame.
\begin{algorithm}
\label{Algo2}
\caption{Backward-pass for $\widehat{\tau}_p, \widehat{\psi}_p$}
\begin{algorithmic} 
\STATE Obtain $\widehat{\tau}_p, \widehat{\psi}_p$ for $p \in \mathcal{P}$ from Algorithm \ref{Algo1} with \eqref{eqn_synch_error_est} or \eqref{eqn_CPE_timing_simpl}.
\FOR{$p=\lfloor P/2 \rfloor:(-1):0$}
\STATE Redefine $\mathcal{Q}_p = \{q \in \mathcal{P} | q > P/2\}$
\STATE Re-calculate and save $\widehat{\tau}_p, \widehat{\psi}_p$ from \eqref{eqn_synch_error_est} or \eqref{eqn_CPE_timing_simpl}.
\ENDFOR
\STATE{// For \eqref{eqn_CPE_timing_simpl} we use $\bar{\mathcal{K}} = \{k \in \mathcal{K} \big| {|\bar{b}_k|}^2 > 0.1 \}$}
\end{algorithmic}
\end{algorithm}
Finally, using either Algorithm \ref{Algo0} or Algorithm \ref{Algo1} or using both Algorithms \ref{Algo1} and \ref{Algo2}, we obtain the CPE and timing error estimates: $\widehat{\tau}_p, \widehat{\psi}_p$ for all $p \in \mathcal{P}, k \in \mathcal{K}$ and compute the corrected CSI $\widehat{{{h}}}_{p,k}$ as in  \eqref{eqn_Hk_synchronized} for performing the sensing operations. 

\section{Evaluation Results} \label{sec_eval_results}
For evaluations, we use stationary WiFi TX and RX operating on a $20$MHz channel at $f_{\rm c} = 5.2$GHz (channel $40$ of the $5$ GHz ISM band). The transmission is via OFDM modulation with a symbol duration $T_{\rm s} = 3.2 \mu$s, cyclic prefix duration $T_{\rm cp}=0.8\mu$s and either $K=64$ sub-carriers (for IEEE 802.11ac CSI) or $K=256$ sub-carriers (for 802.11ax CSI). For the evaluations, we use both a combination of simulated CSI data and real-world CSI data, as explained below. 

\subsection{Distribution of gain and phase errors} \label{subsec_CSI_gain_phase_err_dist}
We first present an analysis of CSI gain and phase errors in a real-world WiFi system under a fully static channel. For this, we use a single-antenna TX that transmits NDP frames every $T_{\rm rep}=50$ms to an RX. We consider two types of RXs: (i) an Intel AX 210 WiFi module and (ii) a Google Nexus 6 phone with a BCM4358 WiFi chip, each with two antennas. 

First, we present a sample measurement of the uncompensated CSI powers $\Gamma_{p}$ (see \eqref{eqn_gamma_define}) for the two RXs in Figs. \ref{Fig_gain_intel} and \ref{Fig_gain_nexus}. For same measurement, we also plot the incremental CSI power  $\Delta \widetilde{\Gamma}_{p}$ (see \eqref{eqn_delta_gamma_defn}) in Figs.~\ref{Fig_gainDiff_intel} and \ref{Fig_gainDiff_nexus}. As can be seen from the results, the clusters in $\Delta \widetilde{\Gamma}_{p}$ are much more prominent than $\widetilde{\Gamma}_{p}$. This is due to the large-scale gain (see Section \ref{subsec_AGC_gain_dist}). Secondly, we also observe from Fig.~\ref{Fig_gainDiff_intel} that the AGC gain indeed increases in multiples of a step size $\lambda$ (as exploited in Section \ref{subsec_gain_uniformG}). 
\begin{figure}[!htb]
\centering
\subfloat[$\widetilde{\Gamma}_{p}$, Intel AX210]{\includegraphics[width= 0.248\textwidth]{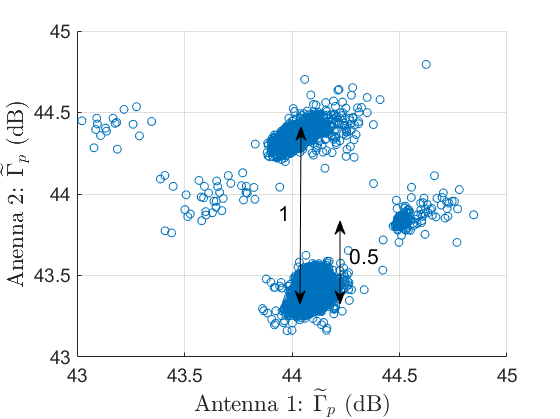} \label{Fig_gain_intel}} 
\subfloat[$\widetilde{\Gamma}_{p}$, BCM 4358]{\includegraphics[width= 0.248\textwidth]{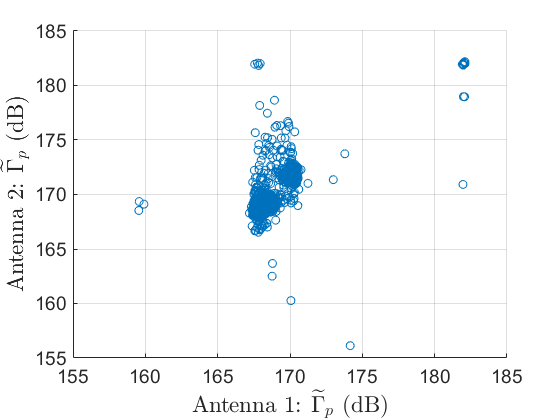} \label{Fig_gain_nexus}} \\
\subfloat[$\Delta \widetilde{\Gamma}_{p}$, Intel AX210]{\includegraphics[width= 0.248\textwidth]{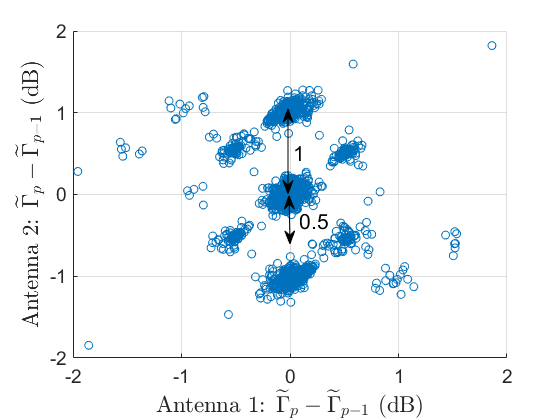} \label{Fig_gainDiff_intel}}
\subfloat[$\Delta \widetilde{\Gamma}_{p}$, BCM 4358]{\includegraphics[width= 0.248\textwidth]{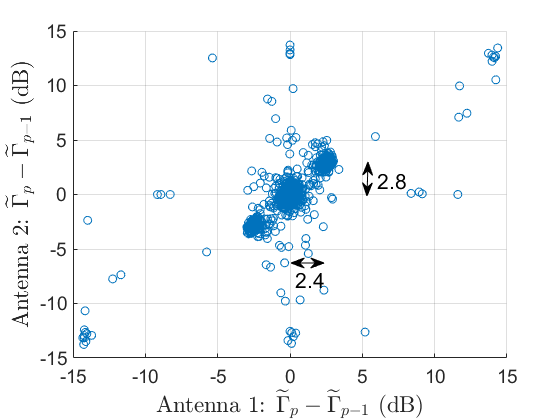} \label{Fig_gainDiff_nexus}}
\caption{Scatter plot of CSI power $\Gamma_{p}$ and incremental power $\Delta \Gamma_{p}$ for the two WiFi RXs in an example static channel.}
\label{Fig_gain_statistics}
\end{figure}

Since it is hard to know the ground truth values phase errors of $\tau_p, \psi_p$ between a TX and RX, here we evaluate the distribution of $\widehat{\tau}_p, \widehat{\psi}_p$ obtained using Algorithm \ref{Algo2}, with \eqref{eqn_CPE_timing_simpl}.\footnote{The results have also been found to be similar when using the other phase error estimation algorithms.} To ensure the estimates are accurate, i.e., $ \widehat{\tau}_p \approx \tau_p, \widehat{\psi}_p \approx \psi_p$, the measurements are made in a high signal-to-noise ratio (SNR), strongly LoS channel. The marginal distributions of $\widehat{\tau}_p$ and $\widehat{\psi}_p$ are depicted in Figs.~\ref{Fig_timing_pdf} and \ref{Fig_phase_pdf}, which verifies that the assumption of uniform marginal distributions for $\widehat{\tau}_p$ and ${\psi}_p$ is reasonably accurate. The cross-covariance coefficients of $\widehat{\tau}_p, e^{{\rm j}\widehat{\psi}_p}$ for two CSI frames $p$ and $p-a$ are depicted in Figs.~\ref{Fig_timing_crosscov} and ~\ref{Fig_phase_crosscov}, as a function of $a$. As can be seen from the results, the timing and phase errors for two CSI frames are uncorrelated for $T_{\rm rep} \geq 50$ms, which justifies the modeling used in Section \ref{subsec_timing_CFO_dist}. 
\begin{figure}[!htb]
\centering
\subfloat[Timing error $\widehat{\tau}_p$]{\includegraphics[width= 0.248\textwidth]{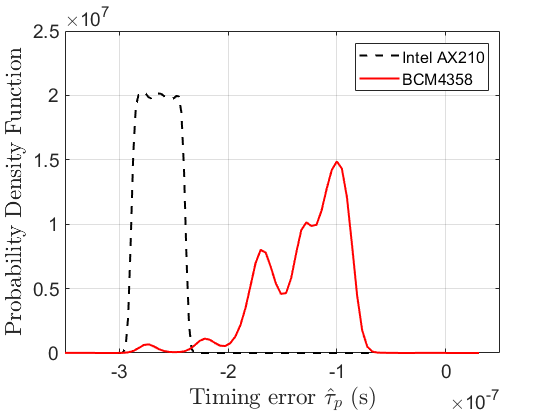} \label{Fig_timing_pdf}} 
\subfloat[Phase error $\widehat{\psi}_p$]{\includegraphics[width= 0.248\textwidth]{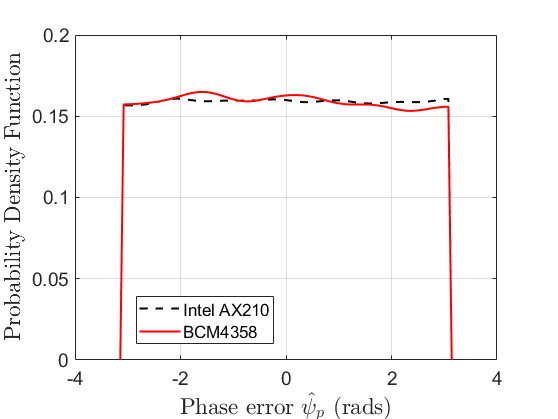} \label{Fig_phase_pdf}} \\
\subfloat[Timing error $\widehat{\tau}_p$]{\includegraphics[width= 0.248\textwidth]{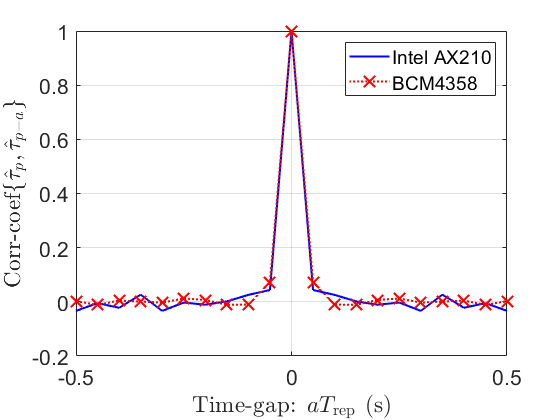} \label{Fig_timing_crosscov}} 
\subfloat[Phase error $\widehat{\psi}_p$]{\includegraphics[width= 0.248\textwidth]{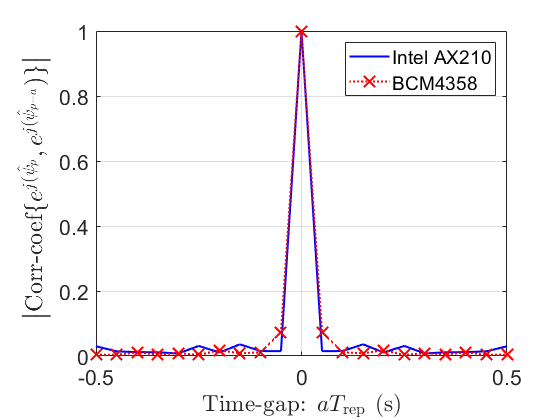} \label{Fig_phase_crosscov}}
\caption{Marginal probability density function and auto-correlation function of $\widehat{\tau}_p$ and $\widehat{\psi}_p$ (using Algorithm \ref{Algo2} with \eqref{eqn_CPE_timing_simpl}), for an example high SNR, static channel.}
\label{Fig_PDF_phase_timing}
\end{figure}

\subsection{Estimation performance in simulated channels} \label{subsec_eval_simulated}
Since knowing the true values of $g_p, \tau_p, \psi_p$ between a TX and RX is difficult in a real-world setting, we compare the performance of the different algorithms over $2000$ realizations of simulated channels. For the simulations, we use $T_{\rm rep} = 100$ms, $\mathcal{P} = \{0,1,...,299\}$ and $\mathcal{K} = \{0,1,...,255\}$. In each realization, $h_{p,k}$ is generated using \eqref{eqn_channel_resp}, where the static component $b_{k}$ is generated using the 802.11ax Model-C channel model, and for the dynamic component we consider 2 cases: 
\begin{enumerate}
\item[(i)] $d_{p,k} \sim \mathcal{CN}(0, 1-\gamma)$ is i.i.d. for each $p,k$. 
\item[(ii)] $d_{p,k} = \alpha_p e^{-{\rm j} 2 \pi f_k (\tau+\tau_0)}$, where $\tau \sim \mathrm{Uni}[0,300ns)$, $\tau_0$ is delay of LoS path in $b_{k}$ and $\alpha_p$ is a complex Gaussian process of variance $1-\gamma$ with the Doppler power spectrum having a non-zero support on $[0.5, 1]$ Hz. 
\end{enumerate}
Among impairments, for the large-scale gain, we use $g^{(1)}_p$ as a real Gaussian process with standard deviation $0.2$ dB and a Doppler power spectrum with non-zero support on $[0, 0.1]$ Hz. For AGC gain, we use $g^{(2)}_p \in \{-0.5, 0, 0.5\}$ with probabilities of $\{0.2, 0.6, 0.2\}$, respectively. For phase errors, we model $\tau_p, \psi_p$ to be i.i.d. for each $p \in \mathcal{P}$ with marginal distributions $\tau_p \sim \mathrm{Uni}[0, {10}^{-7})$ and $\psi_p \sim \mathrm{Uni}[-\pi,\pi)$.
For quantifying the performance of the algorithms, we define the correlation coefficient between $d_{p,k}$ and the dynamic component of cleaned CSI $\widehat{h}_{p,k}$ as:
\begin{eqnarray}
\chi \triangleq \frac{{\Big| \sum_{p \in \mathcal{P}} \sum_{k \in \mathcal{K}} {\big(\widehat{h}_{p,k} - \widehat{b}_k \big)}^{*} d_{p,k} e^{{\rm j} 2 \pi f_k \tau} \Big|}^2}{(1-\gamma) K P{\Big(\sum_{p \in \mathcal{P}} \sum_{k \in \mathcal{K}} {\big|\widehat{h}_{p,k} - \widehat{b}_k \big|}^2 \Big)} } ,
\end{eqnarray}
where $\widehat{b}_k \triangleq \sum_{p \in \mathcal{P}} \widehat{h}_{p,k} / P$ is the static component of cleaned CSI and $$\tau = \max_{\tau} \bigg| \sum_k b_k \widehat{b}_k^{*} e^{{\rm j} 2 \pi f_k \tau} \bigg|,$$ is the timing offset between the true CSI and cleaned CSI. Note that $\chi^2/(1-\chi^2)$ is representative of the SNR in estimating $d_{p,k}$ from cleaned CSI $\widehat{h}_{p,k}$.

The post-cleaning SNR for different gain error correction methods under ideal phase error compensation, i.e., $\widehat{\tau}_p = \tau_p$ and $\widehat{\psi}_p = \psi_p$ is depicted in Fig.~\ref{Fig_compareGainAlgos_vs_gamma} and their computation time is tabulated in Table~\ref{Table_compute_time}. 
As can be seen, at $\gamma=0.9$, Algorithm \ref{Algo4} outperforms both baseline methods by $\geq 100$\% (for ${d}_{p,k}$ of type (i)) and by $40$\% (for ${d}_{p,k}$ of type (ii)). However, for $\gamma > 0.95$, we observe that \eqref{eqn_baseline_normPower} yields the best performance. This is because although the estimation of $g^{(2)}_p$ is accurate, the residual error in estimating $g^{(1)}_p$ in Algorithm \ref{Algo4} dominates over the sensing signal for $\gamma > 0.95$. Under such conditions, completely ignoring CSI power variations (as done by \eqref{eqn_baseline_normPower}) yields the best performance.\footnote{The exact transition point $\gamma = 0.95$ is of course subjective, and increases as we include additive channel noise to \eqref{eqn_H_k_est}.} 
We also observe that Algorithm \ref{Algo4} has a slightly lower computation time than \eqref{eqn_baseline_DBscan} and Algorithm \ref{Algo5}, since its complexity only scales linearly with $P$, as shown in Table \ref{Table_compute_time}. 
\begin{figure}[!htb]
\centering
\subfloat[${d}_{p,k}$ of type (i)]{\includegraphics[width= 0.45\textwidth]{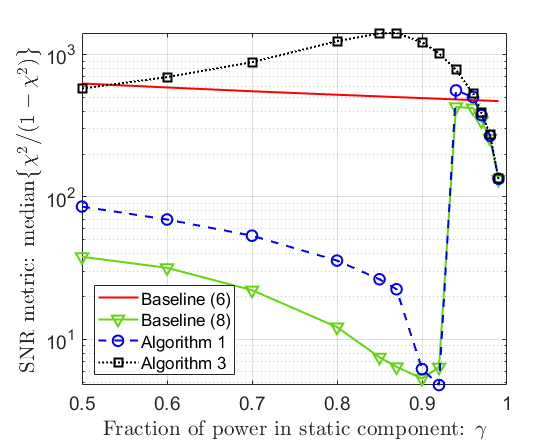} \label{Fig_dynType1_gainAlgos}} \\
\subfloat[${d}_{p,k}$ of type (ii)]{\includegraphics[width= 0.45\textwidth]{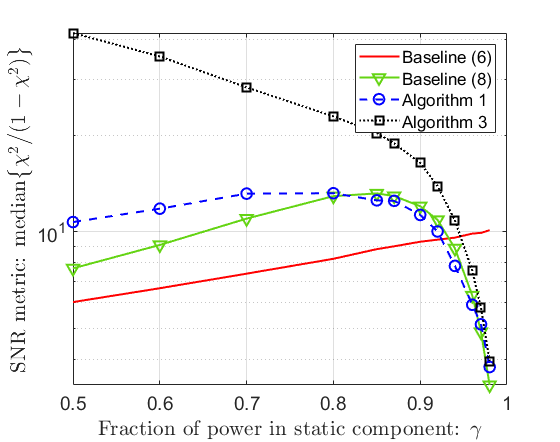} \label{Fig_dynType5_gainAlgos}}
\caption{Median SNR for estimating the dynamic component ${d}_{p,k}$ from cleaned CSI $\widehat{h}_{p,k}$, where $g_p$ is corrected using different algorithms, and with ideal correction of $\tau_p,\psi_p$.}
\label{Fig_compareGainAlgos_vs_gamma}
\end{figure}

The post-cleaning SNR of different phase error correction methods under ideal gain error compensation, i.e., $\widehat{g}_p = g_p$ is depicted in Fig.~\ref{Fig_comparePhaseAlgos_vs_gamma} and their computation time is tabulated in Table~\ref{Table_compute_time}. As can be seen from the results, all the proposed algorithms outperform the baseline methods \eqref{eqn_CPE_timing_baseline1} and \eqref{eqn_CPE_timing_baseline2}. Among the proposed algorithms, the ML estimators \eqref{eqn_timing_err_est_method1} and \eqref{eqn_synch_error_est} yield the best performance\footnote{The dip in their performance for $\gamma \geq 0.95$ is due to finite search resolution for $\tau$ in \eqref{eqn_timing_err_est_method1} and \eqref{eqn_synch_error_est}.},
but they are too computationally cumbersome as shown in Table~\ref{Table_compute_time}. Taking into consideration both performance and computation speed, we conclude that Algorithm \ref{Algo0} with \eqref{eqn_CPE_timing_simpl_method1} or Algorithm \ref{Algo1} with \eqref{eqn_CPE_timing_simpl} yield the best trade-off between performance and computation time. In fact at $\gamma = 0.9$, their performance is at par with the ML estimators, and they outperform both baselines by $> 1000$\% (for ${d}_{p,k}$ of type (i)) and $> 200$\% (for ${d}_{p,k}$ of type (ii)), respectively.
\begin{figure}[!htb]
\centering
\subfloat[${d}_{p,k}$ of type (i)]{\includegraphics[width= 0.45\textwidth]{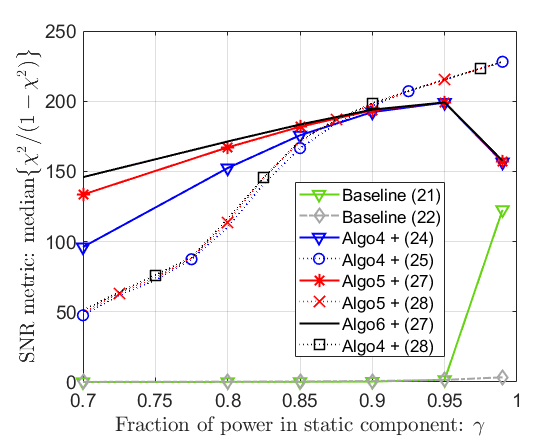} \label{Fig_dynType1_phaseAlgos}} \\
\subfloat[${d}_{p,k}$ of type (ii)]{\includegraphics[width= 0.45\textwidth]{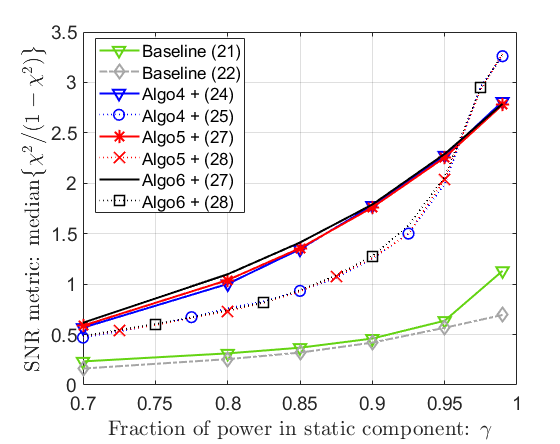} \label{Fig_dynType5_phaseAlgos}}
\caption{Median SNR for estimating the dynamic component ${d}_{p,k}$ from cleaned CSI $\widehat{h}_{p,k}$, where $\tau_p, \psi_p$ are corrected using different algorithms, and with ideal correction of $g_p$.}
\label{Fig_comparePhaseAlgos_vs_gamma}
\end{figure}
\begin{table}[h!]
\centering
\caption{Computation time (in msec) per batch of $P=300$ CSI frames and computation complexity for the proposed methods.} \label{Table_compute_time}
\renewcommand{\arraystretch}{1.5} 
\begin{tabular}{| l | c | c |} 
 \hline 
Methods & Mean Time (ms) & Computation Complexity\tablefootnote{For computation complexity, we use $I$ to indicate number of iterations where applicable.} \\
\hline
Baseline: \eqref{eqn_baseline_normPower} & \color[rgb]{0,0,1} $4.27$ & $\mathrm{O}(PK)$\\
\hline
Baseline: \eqref{eqn_baseline_DBscan} & \color[rgb]{0.1,0,0.9} $12.65$ & $\mathrm{O}(PK + P \log P)$ \\
\hline
Algorithm \ref{Algo5} & \color[rgb]{0.1,0,0.9} $13.75$ & $\mathrm{O}(PK + P \log P)$\\
\hline
Algorithm \ref{Algo4} & \color[rgb]{0.1,0,0.9} $10.98$ & $\mathrm{O}(PK + P I)$ \\
\hline
Baseline: \eqref{eqn_CPE_timing_baseline1} & \color[rgb]{0.2,0,0.8} $24.7$ & $\mathrm{O}(PK)$ \\
\hline
Baseline: \eqref{eqn_CPE_timing_baseline2} & \color[rgb]{0.1,0,0.9} $9.0$ & $\mathrm{O}(PK)$ \\
\hline
Algorithm \ref{Algo0} + \eqref{eqn_timing_err_est_method1} & \color[rgb]{0.7,0,0.3} $443.6$ & $\mathrm{O}(PKI)$ \\
\hline
Algorithm \ref{Algo0} + \eqref{eqn_CPE_timing_simpl_method1} & \color[rgb]{0.3,0,0.7} $40.6$ & $\mathrm{O}(PK)$ \\
\hline
Algorithm \ref{Algo1} + \eqref{eqn_synch_error_est} & \color[rgb]{0.7,0,0.3} $450.5$ & $\mathrm{O}(PKI)$ \\
\hline
Algorithm \ref{Algo1} + \eqref{eqn_CPE_timing_simpl} & \color[rgb]{0.3,0,0.7} $44.5$ & $\mathrm{O}(PK)$ \\
\hline
Algorithm \ref{Algo2} + \eqref{eqn_synch_error_est} & \color[rgb]{0.9,0,0.1} $669.2$ & $\mathrm{O}(PKI)$ \\
\hline
Algorithm \ref{Algo2} + \eqref{eqn_CPE_timing_simpl} & \color[rgb]{0.4,0,0.6} $61.0$ & $\mathrm{O}(PK)$ \\
\hline
\end{tabular}
\end{table}

\subsection{Application to respiration rate estimation} \label{subsec_eval_real}
\begin{figure}[!htb]
\centering
\includegraphics[width= 0.45\textwidth]{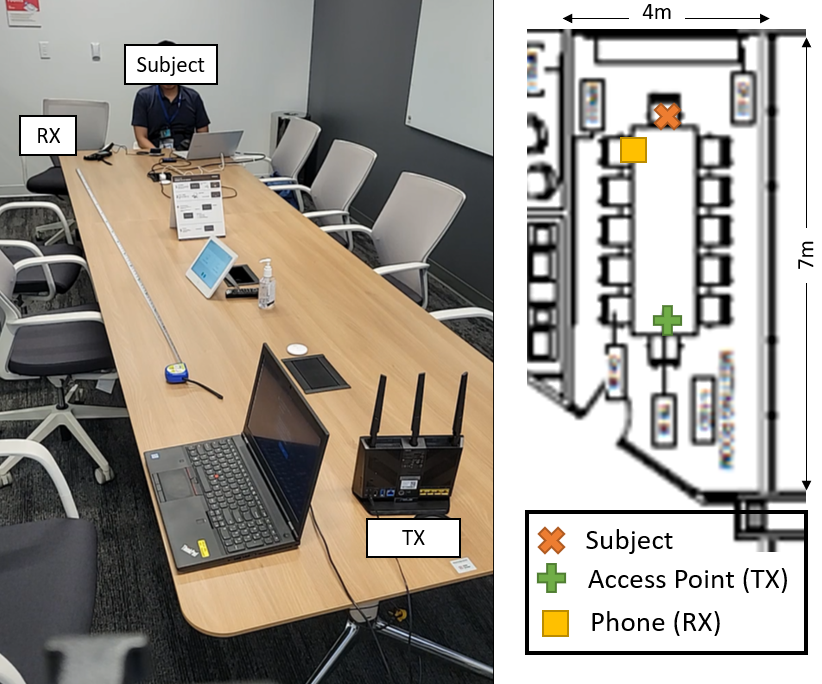}
\caption{Image of measurement setup used to collect the respiration rate data for one example episode}
\label{Fig_respiration_setup}
\end{figure}
In this section, we validate the performance of the proposed CSI preprocessing algorithms in a real-world sensing application of respiration rate monitoring. In this experiment, a stationary WiFi TX and RX are setup in a room along with a subject whose respiration rate has to be monitored. The dimensions of the room are $7{\rm m} \times 4 {\rm m}$, and an illustration of the setup is depicted in Fig.~\ref{Fig_respiration_setup}. The TX transmits WiFi packets (beacons) on channel $155$ of the $5$ GHz ISM band with $20$ MHz bandwidth and $K=64$ sub-carriers at intervals of $T_{\rm rep} = 100$ms. In each measurement episode $(i)$, $P=500$ CSI samples are collected (i.e., $50$ secs). In total, $75$ such episodes are considered, spanning different configurations of TX, RX, and subject locations. In each episode, the obtained CSI $\{\widetilde{h}_{p,k} | p \in \mathcal{P}, k \in \mathcal{K} \}$ is preprocessed using one of the gain and phase correction methods to obtain $\{\widehat{h}_{p,k} | p \in \mathcal{P}, k \in \mathcal{K} \}$ and then the Doppler power spectrum is estimated as:
\begin{eqnarray} \label{eqn_doppler_pow_metric}
\mathscr{H}(\nu) = \sum_{k \in \mathcal{K}} {\bigg| \sum_{p \in \mathcal{P}} \widehat{h}_{p,k} e^{-{\rm j} 2 \pi \nu p T_{\rm rep}} \bigg|}^2 , 
\end{eqnarray}
for $\nu = \{0.1:0.02:0.5\}$ Hz. For each episode $(i)$, the estimation SNR is computed as: 
\begin{eqnarray}
\mathrm{SNR}(i) = \Bigg( \sum_{|\nu - \nu_0| \leq 0.02} \mathscr{H}(\nu) \Bigg) \Bigg/ \Bigg( \sum_{|\nu - \nu_0| > 0.02} \mathscr{H}(\nu) \Bigg) , \nonumber
\end{eqnarray}
where $\nu_{0}$ is the ground-truth respiration rate (in Hz) of the subject as measured with a force belt. The median SNR across the $75$ episodes for different combinations of gain and phase compensation methods is tabulated in Table~\ref{Table_breathing}. 
\begin{table}[h!]
\centering
\caption{Breathing rate SNR under different CSI preprocessing methods.} \label{Table_breathing}
\renewcommand{\arraystretch}{1.5} 
\begin{tabular}{| c | c | c | c | c |} 
\hline 
\textbf{Median SNR} & \multicolumn{4}{ c |}{ Gain correction methods $\rightarrow$ } \\
\hline 
Phase method $\downarrow$ & Eqn.\eqref{eqn_baseline_normPower} & Eqn.\eqref{eqn_baseline_DBscan} & Algo.\ref{Algo5} & Algo.\ref{Algo4} \\
\hline
Eqn.\eqref{eqn_CPE_timing_baseline1} & \color[rgb]{0.9,0,0.1} $0.591$ & \color[rgb]{0.8,0,0.2} $0.611$ & \color[rgb]{0.8,0,0.2} $0.614$ & \color[rgb]{0.8,0,0.2} $0.615$ \\
\hline
Eqn.\eqref{eqn_CPE_timing_baseline2} & \color[rgb]{0.6,0,0.4} $0.663$ & \color[rgb]{0.7,0,0.3} $0.633$ & \color[rgb]{0.8,0,0.2} $0.616$ & \color[rgb]{0.6,0,0.4} $0.634$ \\
\hline
Algo.\ref{Algo0} + \eqref{eqn_timing_err_est_method1} & \color[rgb]{0.2,0,0.8} $0.795$ & \color[rgb]{0.6,0,0.4} $0.674$ & \color[rgb]{0.6,0,0.4} $0.651$ & \color[rgb]{0.3,0,0.7} $0.737$ \\
\hline
Algo.\ref{Algo0} + \eqref{eqn_CPE_timing_simpl_method1} & \color[rgb]{0,0,1} $0.799$ & \color[rgb]{0.6,0,0.4} $0.675$ & \color[rgb]{0.6,0,0.4} $0.652$ & \color[rgb]{0.3,0,0.7} $0.739$ \\
\hline
Algo.\ref{Algo1} + \eqref{eqn_synch_error_est} & \color[rgb]{0.2,0,0.8} $0.795$ & \color[rgb]{0.6,0,0.4} $0.674$ & \color[rgb]{0.6,0,0.4} $0.651$ & \color[rgb]{0.3,0,0.7} $0.738$ \\
\hline
Algo.\ref{Algo1} + \eqref{eqn_CPE_timing_simpl} & \color[rgb]{0,0,1} $0.799$ & \color[rgb]{0.6,0,0.4} $0.675$ & \color[rgb]{0.6,0,0.4} $0.652$ & \color[rgb]{0.3,0,0.7} $0.739$ \\ 
\hline
Algo.\ref{Algo2} + \eqref{eqn_synch_error_est} & \color[rgb]{0.2,0,0.8} $0.795$ & \color[rgb]{0.6,0,0.4} $0.674$ & \color[rgb]{0.6,0,0.4} $0.651$ & \color[rgb]{0.3,0,0.7} $0.737$ \\
\hline
Algo.\ref{Algo2} + \eqref{eqn_CPE_timing_simpl} & \color[rgb]{0,0,1} $0.799$ & \color[rgb]{0.6,0,0.4} $0.675$ & \color[rgb]{0.6,0,0.4} $0.652$ & \color[rgb]{0.3,0,0.7} $0.739$ \\
\hline
\end{tabular}
\end{table}
As can be seen from the results, the proposed phase-correction methods beat both the baseline methods by $> 20$\% in median SNR. However, for gain correction, we see that the baseline method \eqref{eqn_baseline_normPower} yields the best performance, followed by our proposed method in Algorithm \ref{Algo4}. The superior performance of \eqref{eqn_baseline_normPower} is because for this example sensing task we have $\gamma > 0.99$ where, as also observed in Section \ref{subsec_eval_simulated}, ignoring CSI power variation yields the best performance. Thus, these observations are well aligned with our simulation results. 

\section{Conclusion and Future directions} \label{sec_conclusions}
This paper develops a mathematical model for the gain and phase errors in WiFi CSI and proposes several algorithms for resolving the gain and phase errors. The analysis of the real-world data shows that the gain errors can be treated as the combination of a slow-changing large-scale gain and a fast-changing AGC gain that attains a value from a discrete uniform grid. It is shown that for the true channel, the CSI power is approximately Gaussian distributed in the dB scale, and its ML estimate from the observed CSI is wrapped Gaussian distributed, due to the ambiguity introduced by the AGC gain. This can be leveraged to remove the gain errors, as shown in Algorithm \ref{Algo4}. 
The analysis of real-world data also shows that the CSI phase errors for packets spaced $\geq 50$ms apart are independently distributed. It is also shown that the conditional ML estimation of the phase errors can be simplified to a linear weighted least-squares estimation problem under certain simplifying assumptions. 
Simulations under standard-compliant channel models show that when the sensing signal is not i.i.d., cleaning the CSI is more difficult, and hence, choosing a good preprocessing algorithm is important. Results also show that the proposed theoretically justified gain-correction and phase-correction algorithms, achieve $40$\% and $200$\% higher estimation SNR, respectively, compared to baseline methods at $\gamma=0.9$, without significantly increasing computation complexity. We also conclude that when the sensing signal is very weak, i.e., $\gamma > 0.95$, ignoring the gain variation yields the best performance.
The proposed algorithms are also applied in a real-world test bed for performing respiration rate monitoring using CSI. The results show that the proposed phase-correction methods can indeed improve the estimation SNR by $20$\% compared to baseline methods, while for gain-correction, the method in \eqref{eqn_baseline_normPower} works best (since we have $\gamma > 0.99$ in this test).

There is significant scope for further work on CSI preprocessing. Firstly, the proposed CPE model ignored its relationship with the CFO. Once the steps of Section \ref{sec_CPE_timing_correct} are done, a further refinement step to $\widehat{\psi}_p$ can be envisioned by exploiting the relationship. Secondly, although the proposed pre-processing approaches can be used in a multi-antenna setting (by performing CSI cleaning per antenna pair) this may be sub-optimal. Using the fact that the timing errors $\tau_p$ are shared by all antenna pairs and the CPE $\psi_p$ is correlated across antennas, better joint estimation approaches can be envisioned. Finally, due to variation in the TX oscillator frequency and other RF impairments, the dominant channel component $b_k$ may slowly change over time. This change is non-negligible when $P$ is large, and thus algorithms that can accommodate slowly varying $b_k$ can be devised. 

\begin{appendices}
\section{} \label{appdix1p5}
\begin{proof}[Proof of Lemma \ref{remark1}]
From \eqref{eqn_channel_resp}, note that:
\begin{eqnarray}
\beta_p \triangleq \frac{2}{(1-\gamma)} \sum_{k \in \mathcal{K}} {| {b}_{k} + {d}_{p,k} |}^2,
\end{eqnarray}
is a standard, non-central Chi-squared distributed random variable with a non-centrality parameter of $\frac{2 K \gamma}{(1-\gamma)}$ and $2K$ degrees of freedom, and is independent for each $p$. Correspondingly, its mean and variance are:
\begin{eqnarray}
\mu_{\beta} = \frac{2K}{1-\gamma}, \ \sigma_{\beta}^2 = 4K \frac{1 + \gamma}{1-\gamma}. \nonumber
\end{eqnarray}
In other words, variance of $(1-\gamma) \beta_p \rightarrow 0$ as $\gamma \rightarrow 1$. Correspondingly from \eqref{eqn_gamma_true_define} we have:
\begin{eqnarray}
\Gamma_p &=& 10 \times \log_{10} \big[(1-\gamma) \beta_p \big/ 2K \big] \nonumber \\
&\approx & \frac{5 [ (1-\gamma) \beta_p - 2K]}{K} \nonumber \\
&=&  10\sqrt{\frac{1 - \gamma^2}{K}}\frac{(\beta_p - \mu_{\beta})}{\sigma_{\beta}} \label{eqn_gamma_temp} 
\end{eqnarray}
where the second step follows by using a first-order Taylor expansion about $\beta_p = \mu_p$, which is accurate for $\gamma \approx 1$. It is well known that either as $2K \rightarrow \infty$ or $2 K \gamma \big/ (1-\gamma) \rightarrow \infty$, $(\beta_p - \mu_{\beta})/\sigma_{\beta} \rightarrow \mathcal{N}(0,1)$ \cite{Muirhead1982}. In our case since number of sub-carriers $K \gg 1$ and also $\gamma \approx 1$, both conditions for convergence of $\beta_p$ are satisfied and hence from \eqref{eqn_gamma_temp} we have that $\Gamma_p$ is approximately Gaussian distributed with a mean of $0$ and a small variance $100 (1 - \gamma^2)\big/K$.
\end{proof}

\begin{proof}[Proof of Lemma \ref{remark2}]
Note that from \eqref{eqn_H_k_est}, \eqref{eqn_gamma_define}, \eqref{eqn_gamma_true_define} and the AGC gain model, we can write: 
\begin{eqnarray}
\widetilde{\Gamma}_{p} = {\Gamma}_{p} + g^{(1)}_{p} + z_p \lambda. \nonumber
\end{eqnarray}
Since $z_p$ can be any integer with a uniform prior, we have 
$$\mathbb{P} \big( \widetilde{\Gamma}_{p} = g \big| {\Gamma}_{p} \big) = \mathbb{P} \big(\widetilde{\Gamma}_{p} = g + \lambda \big| {\Gamma}_{p} \big), $$
and the result follows.
\end{proof}

\begin{proof}[Proof of Lemma \ref{remark3}]
Using \eqref{eqn_H_k_est}, \eqref{eqn_gamma_define} and \eqref{eqn_gamma_true_define} in \eqref{eqn_gamma_hat} we can obtain:
\begin{equation} \label{eqn_gamma_hat_simplify}
\widehat{\Gamma}_{p} = \lambda \big \lfloor {\Gamma}_{p} \big/ \lambda \big\rfloor .
\end{equation}
Since ${\Gamma}_{p}$ is zero mean, Gaussian distributed (from Lemma \ref{remark1}), it follows that $\widehat{\Gamma}_{p}$ is wrapped Gaussian distributed (see Fig.~\ref{Fig_gamma_distortion}). Correspondingly, using the properties of wrapped Gaussian distribution \cite{mardia1999directional}, it follows that $\mathbb{E} \{ \widehat{\Gamma}_{p} \} = e^{-\sigma_{\Gamma}^2/2}$. Finally using the sample mean estimate of $\mathbb{E} \{ \widehat{\Gamma}_{p} \}$, the result follows.
\end{proof}

\begin{proof}[Proof of Lemma \ref{remark4}]
The distortion between ${\Gamma}_{p}, \widehat{\Gamma}_{p}$ is given by:
\begin{multline*}
\int_{-\infty}^{\infty} \mathbb{P} \big( {\Gamma}_{p} = g \big) {\Big( \widehat{\Gamma}_{p} - {\Gamma}_{p} \Big)}^2 {\rm d} g \\
= \sum_{z \in \mathbb{Z}} \int_{-\frac{\lambda}{2}}^{\frac{\lambda}{2}} \mathbb{P} \big( {\Gamma}_{p} = g + z \lambda \big) {\big( z \lambda \big)}^2 {\rm d} g , \nonumber
\end{multline*} 
which follows from \eqref{eqn_gamma_hat_simplify}. Finally, since ${\Gamma}_{p}$ is Gaussian distributed with a variance of $ \sigma_{\Gamma}^2$ from Lemma \ref{remark1}, the result follows.
\end{proof}

\section{} \label{appdix0}
\begin{proof}[Proof of Lemma \ref{Th_method1_estimation_simpl}]
When $\gamma \approx 1$, from \eqref{eqn_channel_resp}, we have that: $\bar{{{h}}}_{p,k} \approx {{h}}_{p,k} \approx {b}_{k}$ for all $p \in \mathcal{P}$. Additionally, in strongly LoS channels we have $b_k \approx b_{k+a}$ for $|a| \leq 2$, i.e., channel has a high coherence bandwidth. Therefore, using \eqref{eqn_Hk_gain_synchronized}, we can write:
\begin{eqnarray}
\sum_{k=1}^{K-1} \bar{h}_{p,k} \bar{h}_{p,k+1}^{*} \approx \sum_{k=1}^{K-1} {|b_{k}|}^2  e^{{\rm j} \frac{{\rm j} 2 \pi \tau}{T_{\rm s}}}, \label{eqn_remark5_tau_support}
\end{eqnarray}
where we neglected the other terms which vanish as $\gamma \rightarrow 1$. Thus, from \eqref{eqn_remark5_tau_support}, we conclude that \eqref{eqn_CPE_baseline2_tau} can provide a good initial estimate $\bar{\tau}_p$ of $\tau_p$, i.e., $|\tau_p - \bar{\tau}_p|$ is small. Correspondingly we can write:
\begin{eqnarray} \label{eqn_slow_phase_variation}
\bar{h}^{*}_{p,k} \bar{b}_k e^{- {\rm j} 2 \pi f_k \bar{\tau}_p} \approx {|b_k|}^2 e^{{\rm j} (2 \pi f_k (\tau_p-\bar{\tau}_p) + \psi_p)}.
\end{eqnarray}
In other words, the magnitude and phase angle of $\bar{h}^{*}_{p,k} \bar{b}_k e^{- {\rm j} 2 \pi f_k \bar{\tau}_p}$ has a very slow variation with $k$, which shall be exploited below. Note that, we can simplify \eqref{eqn_timing_err_est_method1} as:
\begin{flalign} \label{eqn_timing_err_est_method11}
& \widehat{\tau}_p, \widehat{\psi}_p = \argmax_{|\tau | \leq \frac{\kappa T_{\rm s}}{K}, |\psi| < \pi} \!\! \bigg\{ & \nonumber \\
& \qquad \qquad \sum_{k \in \bar{\mathcal{K}}} \mathrm{Re} \Big\{ e^{-{\rm j} \big(2 \pi f_k \tau + \psi - \angle [ {\bar{{{h}}}_{p,k}}^{*} \bar{b}_k ]\big)} \Big\} \Big| {\bar{{{h}}}_{p,k}}^{*} \bar{b}_k \Big| \bigg\}, \!\!\!\!\! &
\end{flalign}
where $\bar{\mathcal{K}} = \{k \in \mathcal{K} \big| | \bar{b}_k | > 0 \}$. Let us define:
\begin{flalign} \label{eqn_phase_err_defn}
& \bar{\varphi}_{p,k}(\tau, \psi) \triangleq 2 \pi f_k (\tau - \bar{\tau}_p) + \psi - \widetilde{\mathscr{U}}\big[ \angle  \big( {\bar{{{h}}}_{p,k}}^{*} \bar{b}_k e^{- {\rm j} 2 \pi f_k \bar{\tau}_p} \big) \big] & \end{flalign}
where $\widetilde{\mathscr{U}}(\cdot)$ is as defined in the lemma statement. Applying this to \eqref{eqn_timing_err_est_method11}, we get:
\begin{flalign}
& \widehat{\tau}_p, \widehat{\psi}_p = \argmax_{|\tau | \leq \frac{\kappa T_{\rm s}}{K}, |\psi| < \pi} \bigg\{ \sum_{k \in \bar{\mathcal{K}}} \cos \big[ \bar{\varphi}_{p,k}(\tau, \psi) \big] \Big| {\bar{{{h}}}_{p,k}}^{*} \bar{b}_k \Big| \bigg\} \!\!\!\!\! & \label{eqn_ML_est_6}
\end{flalign}
When $\gamma \approx 1$, from \eqref{eqn_channel_resp}, we additionally have that at the true CPE and timing offsets: $\bar{{{h}}}_{p,k} e^{{\rm j} (2 \pi f_k {\tau}_q + \psi_q)} \approx {{{b}}}_{k}$ for all $p \in \mathcal{P}$. Correspondingly for any $k \in \bar{\mathcal{K}}$, we have that $\cos(\bar{\varphi}_{p,k}({\tau}_p, \psi_p)) \approx 1$. Since $({\tau}_p, \psi_p)$ is a feasible solution to \eqref{eqn_ML_est_6}, it follows that $\cos[\bar{\varphi}_{p,k}(\widehat{\tau}_p, \widehat{\psi}_p)] \approx 1$ is also true at the optimal solution $(\widehat{\tau}_p, \widehat{\psi}_p)$ to \eqref{eqn_ML_est_6}. In other words, there exists an integer $z_p \in \mathbb{Z}$ such that for all $k \in \bar{\mathcal{K}}$, $\bar{\varphi}_{p,k}(\widehat{\tau}_p, \widehat{\psi}_p) \approx 2 \pi z_p$. Note that we use the same integer $z_p$ for all sub-carriers $k \in \bar{\mathcal{K}}$ since $\bar{\varphi}_{p,k}(\tau, \psi)$ varies smoothly with $k$ (see \eqref{eqn_slow_phase_variation}) and a transition between two different integers would violate $\cos[\bar{\varphi}_{p,k}(\widehat{\tau}_p, \widehat{\psi}_p)] \approx 1$ for at least some sub-carriers $k \in \bar{\mathcal{K}}$. Thus, since $|\bar{\varphi}_{p,k}(\tau, \psi) - 2 \pi z_p| \ll 1$ at the optimal solution to \eqref{eqn_ML_est_6}, we can reduce \eqref{eqn_ML_est_6} to: 
\begin{flalign} \label{eqn_ML_est_7}
& \argmax_{|\tau | \leq \frac{\kappa T_{\rm s}}{K}, |\psi| < \pi} \bigg\{ \sum_{k \in \bar{\mathcal{K}}} \bigg[1 - \frac{{\big(\bar{\varphi}_{p,k}(\tau, \psi) - 2 \pi z_p\big)}^2}{2} \bigg] \big| {\bar{{{h}}}_{p,k}}^{*} \bar{b}_k \big| \bigg\} . &
\end{flalign}
where we use the fact that for any real $x$:
\begin{eqnarray}
\cos(x) \approx 1 - x^2/2 \ \text{if $|x| \ll 1$} \nonumber \\
\cos(x) \geq 1 - x^2/2 \ \text{otherwise} \nonumber
\end{eqnarray}
Applying \eqref{eqn_phase_err_defn} to \eqref{eqn_ML_est_7} and, without loss of generality, absorbing the $2 \pi z_p$ into $\psi$, \eqref{eqn_CPE_timing_simpl_method1} follows.
\end{proof}

\begin{proof}[Proof of Lemma \ref{Th_ML_estimation}]
For given prior information: $$\star = \{\bar{{{h}}}_{q,k}, \widehat{\tau}_q, \widehat{\psi}_q | q \in \mathcal{Q}_p\}$$ such that $\widehat{\tau}_q = {\tau}_q, \widehat{\psi}_q = {\psi}_p$, we are interested to find the ML estimates:
\begin{flalign} \label{eqn_ML_est_0}
& \widehat{\tau}_p, \widehat{\psi}_p = \argmax_{\tau, \psi} \bigg\{ \max_{\{{b}_{k} | k \in \mathcal{K}\}} \bigg[ \mathbb{P}({\psi}_p = \psi) \mathbb{P}({\tau}_p = \tau) & \nonumber \\
& \qquad \prod_{k \in \mathcal{K}} \bigg[  \mathbb{P} \big( \bar{{{h}}}_{p,k} \big| {\psi}_p = \psi, {\tau}_p = \tau, {{b}}_{k}, \star \big) & \nonumber \\
& \qquad  \Big( \prod_{q \in \mathcal{Q}_p} \mathbb{P} \big( \bar{{{h}}}_{q,k} \big| {{{b}}}_{k}, \star \big) \Big) \mathbb{P}({{b}}_{k}) \bigg] \bigg\}, &
\end{flalign}
where we use the fact that ${{{d}}}_{p,k}$ are independent for each $k$ in \eqref{eqn_channel_resp}. Noting that ${\tau}_p \sim {\rm Uni}(-\frac{\kappa T_{\rm s}}{K}, \frac{\kappa T_{\rm s}}{K})$, ${\psi}_p \sim {\rm Uni}(-\pi, \pi)$ (see Section \ref{subsec_timing_CFO_dist}), ${{{d}}}_{p,k} e^{-{\rm j} {\psi}_p} e^{-{\rm j} 2 \pi f_k \tau_p}$ are i.i.d. complex Gaussian vectors for each $p \in \mathcal{P}$ and $b_k$ has a uniform prior (see \eqref{eqn_channel_resp} and \eqref{eqn_H_k_est}), we can rewrite the right hand side of \eqref{eqn_ML_est_0} as:
\begin{flalign}
& \argmax_{ |\tau | \leq \frac{\kappa T_{\rm s}}{K}, |\psi| < \pi} \Bigg\{ \max_{\sum_k {|{b}_k|}^2 = \gamma K} \prod_{k \in \mathcal{K}} \Bigg[  & \nonumber \\
& \qquad \exp \bigg\{ \frac{- {\big| {b}_{k} e^{-{\rm j} [2 \pi f_k \tau + \psi]} - \bar{{{h}}}_{p,k} \big|}^2}{ 1- \gamma } \bigg\} & \nonumber \\
& \qquad \prod_{q \in \mathcal{Q}_p} \exp \bigg\{ \frac{-{\big| {b}_{k} e^{-{\rm j} [2 \pi f_k \widehat{\tau}_{q} + \widehat{\psi}_{q}]} - \bar{{{h}}}_{q,k} \big|}^2}{ 1- \gamma } \bigg\} \Bigg] \Bigg\} & 
\end{flalign}
Taking the logarithm on both sides, we can can further reduce the optimization problem to:
\begin{flalign} \label{eqn_ML_est_1}
& \argmin_{ |\tau | \leq \frac{\kappa T_{\rm s}}{K}, |\psi| < \pi} \bigg\{ \min_{\sum_k {|b_k|}^2 = K \gamma} \bigg[ \sum_{k \in \mathcal{K}} \bigg( {\big| {b}_{k} e^{-{\rm j} [2 \pi f_k \tau + \psi]} - \bar{{{h}}}_{p,k} \big|}^2 & \nonumber \\
& \quad + \sum_{q \in \mathcal{Q}_p} {\big| {b}_{k} e^{-{\rm j} [2 \pi f_k \widehat{\tau}_{q} + \widehat{\psi}_{q}]} - \bar{{{h}}}_{q,k}  \big|}^2 \bigg) \bigg] \bigg\} \!\!\!\!\!\!\!\!\!\! &
\end{flalign}
The conditional optimal solution of ${b}_{k}$ in \eqref{eqn_ML_est_1} for given $\tau, \psi$ is given by:
\begin{eqnarray} \label{eqn_ML_est_4}
{b}_{k} = A \bigg[ \bar{{{h}}}_{p,k} e^{{\rm j} (2 \pi f_k \tau + \psi)} + \sum_{q \in \mathcal{Q}_p} \widehat{{{h}}}_{q,k} \bigg],
\end{eqnarray}
where we use \eqref{eqn_Hk_synchronized} and  $A$ is such that $\sum_k {|b_k|}^2 = \gamma K$. Substituting \eqref{eqn_ML_est_4} back into \eqref{eqn_ML_est_1} and removing the terms that are independent of $\tau, \psi$ from the objective function, \eqref{eqn_ML_est_1} can be reduced to:
\begin{flalign} \label{eqn_ML_est_5}
& \argmin_{|\tau | \leq \frac{\kappa T_{\rm s}}{K}, |\psi| < \pi} \Bigg\{ & \nonumber \\
& \qquad \quad \sum_{k \in \mathcal{K}} - \mathrm{Re} \bigg\{ e^{-{\rm j} (2 \pi f_k \tau + \psi)} {\bar{{{h}}}_{p,k}}^{*} \Big( \sum_{q \in \mathcal{Q}_p} \widehat{{{h}}}_{q,k} \Big) \bigg\} \Bigg\} \!\!\!\! &
\end{flalign}
Finally re-arranging the terms of \eqref{eqn_ML_est_5} we arrive at \eqref{eqn_timing_err_est} and \eqref{eqn_CPE_est}.
\end{proof}

\end{appendices}




%

\bibliographystyle{ieeetr}
\bibliography{references}

\begin{IEEEbiography}[{\includegraphics[width=1in,height=1.25in,clip,keepaspectratio]{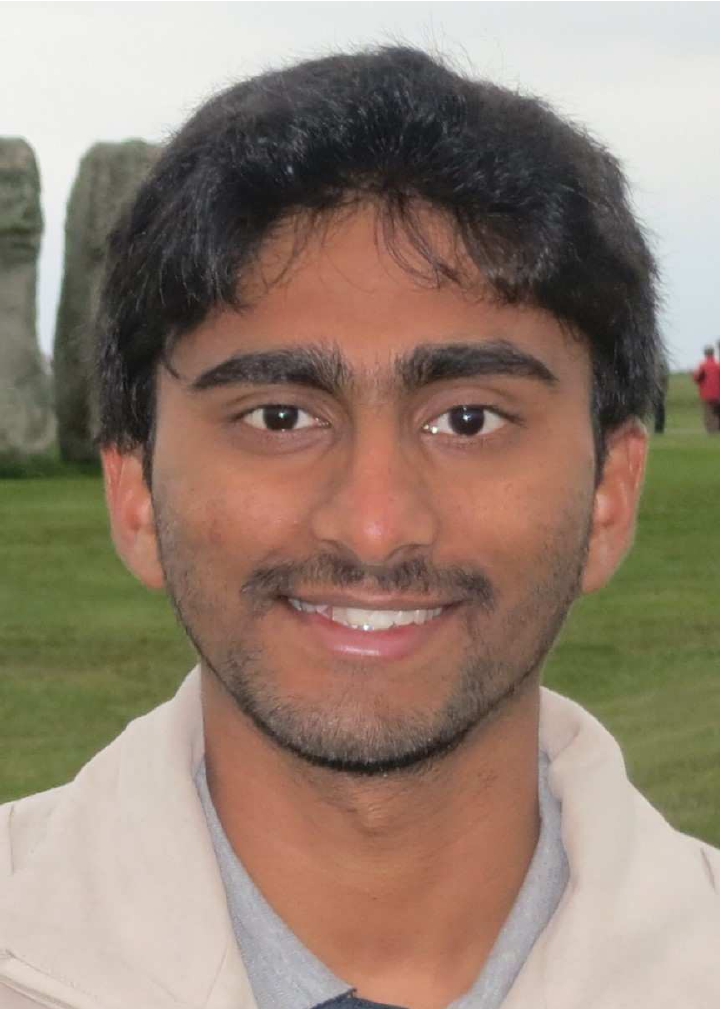}}]{Vishnu V. Ratnam} (S'10--M'19--SM'22) received the B.Tech. degree (Hons.) in electronics and electrical communication engineering from IIT Kharagpur, Kharagpur, India in 2012, where he graduated as the Salutatorian for the class of 2012. He received the Ph.D. degree in electrical engineering from University of Southern California, Los Angeles, CA, USA in 2018. He currently works as a Staff Research Engineer II in the Standards and Mobility Innovation Lab at Samsung Research America, Plano, Texas, USA. His research interests are in Wi-Fi standards, wireless sensing, AI for wireless, mm-Wave and Terahertz communication, massive MIMO, and resource allocation problems in multi-antenna networks. 
Dr. Ratnam was the recipient of the Best Student Paper Award with the IEEE International Conference on Ubiquitous Wireless Broadband (ICUWB) in 2016, the Bigyan Sinha memorial award in 2012 and is a member of the Phi-Kappa-Phi honor society.
\end{IEEEbiography}

\begin{IEEEbiography}[{\includegraphics[width=1in,height=1.25in,clip,keepaspectratio]{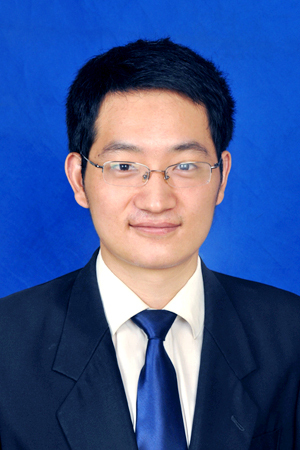}}]{Hao Chen} received his B.S. and M.S. degrees in Information Engineering from Xi'an Jiaotong University, Shaanxi, in 2010 and 2013. He received the Ph.D. degree in Electrical Engineering from University of Kansas, Lawrence, KS, in 2017. He currently works as a Senior Staff Engineer with the Standards and Mobility Innovation Laboratory, Samsung Research America, where he is working on algorithm design and prototyping of AI for wireless communication, wireless sensing, and localization. His research interests include network optimization, machine learning, and 5G cellular systems.
\end{IEEEbiography}

\begin{IEEEbiography}[{\includegraphics[width=1in,height=1.25in,clip,keepaspectratio]{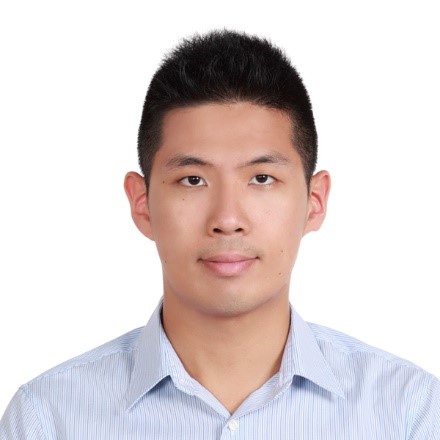}}]{Hao-Hsuan Chang} received the Ph.D. degree in electrical and computer engineering from Virginia Tech, USA, in 2021, the M.S. degree in communication engineering from National Taiwan University, Taiwan, in 2015, and the B.Sc. degree in electrical engineering from National Taiwan University, Taiwan, in 2013. He is currently the Senior Research Engineer in Samsung Research America. His research interests include dynamic spectrum access, wireless sensing, and machine learning for wireless communications.
\end{IEEEbiography}

\begin{IEEEbiography}[{\includegraphics[width=1in,height=1.25in,clip,keepaspectratio]{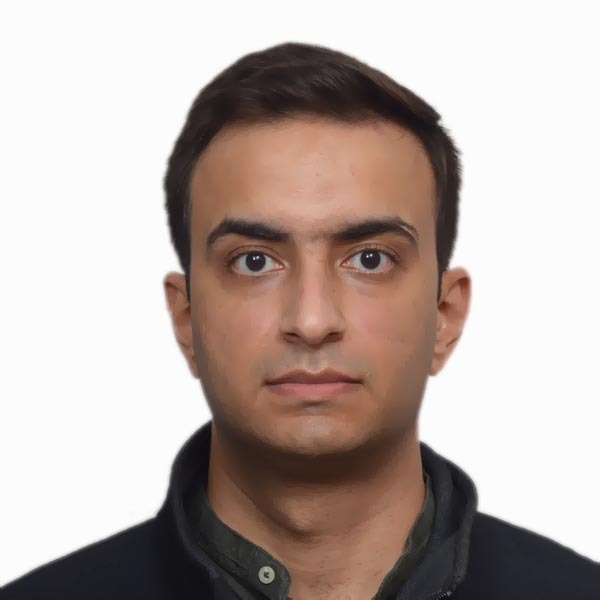}}]{Abhishek Sehgal} (S'15--M'19)  received the B.E. degree in instrumentation technology from Visvesvaraya Technological University, Belgaum, India, in 2012, and the M.S. and Ph.D. degrees in electrical engineering from The University of Texas at Dallas, Richardson, TX, USA, in 2015 and 2019, respectively. He is currently a Staff Research Engineer II with the Standards and Mobility Innovation (SMI) Laboratory, Samsung Research America. His research interests include real-time signal processing, pattern recognition, machine learning, and wireless communication. His awards and honors include the Second Place Award for the Hearables Challenge organized by the National Science Foundation (NSF), in 2017, and is a member of the Phi-Kappa-Phi society.
\end{IEEEbiography}

\begin{IEEEbiography}[{\includegraphics[width=1in,height=1.25in,clip,keepaspectratio]{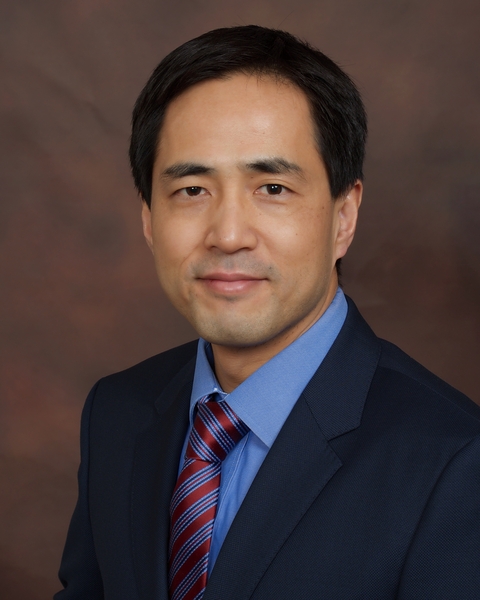}}]{Jianzhong (Charlie) Zhang} (S'00-M'03-SM'09-F'16) received the Ph.D. degree from the University of Wisconsin, Madison WI, USA. He is currently a Senior Vice President at Samsung Research America, where he leads research, prototyping, and standardization for 5G/6G and other wireless systems. He is also a Corporate VP and head of the  global 6G team at Samsung Research. He is currently serving as the ATIS North America Next-G Alliance Full Member Group Vice Chair. Previously, he was the Board Chair of the FiRa Consortium  from May 2019 to May 2023, and the Vice Chairman of the 3GPP RAN1 working group from 2009 to 2013, where he led development of LTE and LTE-Advanced technologies. He worked for Nokia Research Center and Motorola Mobility for 6 years before joining Samsung in 2007.  He received his Ph.D. degree from the University of Wisconsin, Madison.  Dr. Zhang is a Fellow of IEEE.
\end{IEEEbiography}

\end{document}